\documentclass[10pt,journal,compsoc,letterpaper]{IEEEtran}

\usepackage{cite}

\usepackage{fancyhdr}
\usepackage[table]{xcolor}

\usepackage{graphicx}

\usepackage{amsthm}
\usepackage{bm}
\usepackage{amsmath}
\usepackage{amssymb}

\usepackage{tikz}

\theoremstyle{definition}
\newtheorem{definition}{Definition}[section]

\newcommand{\acronym}{OMNI}

\usepackage{subcaption}
\usepackage{multirow}
\usepackage{booktabs}
\usepackage{lineno,hyperref}
\usepackage{algpseudocode}
\usepackage{algorithm}
\algrenewcommand\algorithmicindent{1.2em}%

\usepackage{tabularx,ragged2e}

\usepackage{nicefrac}

\newcommand{\squishlist}{
   \begin{list}{$\bullet$}
    { \setlength{\itemsep}{2pt}    \setlength{\parsep}{0pt}
      \setlength{\topsep}{5pt}     \setlength{\partopsep}{0pt}
      \setlength{\leftmargin}{1.35em} \setlength{\labelwidth}{1em}
      \setlength{\labelsep}{0.5em} } }

\newcommand{\squishend}{
    \end{list}  }

\newtheorem{theorem}{Theorem}
\newtheorem{example}{Example}

\newcommand{\revised}[1]{\textcolor{black}{#1}}

\begin{document}

  \title{Observation-Enhanced QoS Analysis of Component-Based Systems}
  \author{Colin Paterson, Radu Calinescu
  \IEEEcompsocitemizethanks{\IEEEcompsocthanksitem C.~Paterson and R.~Calinescu are with the Department of Computer Science at the University of York, UK.}}


\IEEEtitleabstractindextext{

\begin{abstract}
We present a new method for the accurate analysis of the quality-of-service (QoS) properties of component-based systems. Our method takes as input a QoS property of interest and a high-level continuous-time Markov chain (CTMC) model of the analysed system, and refines this CTMC based on observations of the execution times of the system components. The refined CTMC can then be analysed with existing probabilistic model checkers to accurately predict the value of the QoS property. 
The paper describes the theoretical foundation underlying this model refinement, the tool we developed to automate it, and two case studies that apply our QoS analysis method to a service-based system implemented using public web services and to an IT support system at a large university, respectively.
Our experiments show that traditional CTMC-based QoS analysis can produce highly inaccurate results and may lead to invalid engineering and business decisions. 
In contrast, our new method reduced QoS analysis errors by 84.4--89.6\% for the service-based system and by 94.7--97\% for the IT support system, significantly lowering the risk of such invalid decisions.
\end{abstract}

}

\maketitle


\thispagestyle{fancy}

\section{Introduction}

Modern software and information systems are often constructed using complex interconnected components \cite{leavens2000foundations}. The performance, cost, resource use and other quality-of-service (QoS) properties of these systems underpin important engineering and business decisions. As such, the QoS analysis  of component-based systems has been the subject of intense research \cite{aleti2013software,Becker2006,Koziolek2010634,koziolek2013hybrid}. The solutions devised by this research can analyse a broad range of QoS properties by using \emph{performance models} such as Petri Nets \cite{marsan-etal1994,perez-etal2012}, layered queuing networks \cite{franks-etal2009}, Markov chains \cite{carrasco2002computationally,sato2007stochastic} and timed automata \cite{Hessel2008}, together with tools for their simulation (e.g.\ Palladio \cite{becker-etal2009} and GreatSPN  \cite{baarir2009greatspn}) and formal verification (e.g.\ PRISM \cite{kwiatkowska2011prism} and UPPAAL \cite{1704000}).

These advances enable the effective analysis of many types of performance models. However, they cannot support the design and verification of real systems unless the analysed models are accurate representations of the system behaviour, and ensuring the accuracy of performance models remains a major challenge. Our paper address this challenge for \emph{continuous-time Markov chains} (CTMCs), a type of stochastic state transition models used for QoS analysis at both design time \cite{sato2007stochastic,GallottiGMT2008,7372021} and runtime \cite{FilieriGT12,calinescu2015self}. To this end, we present a tool-supported method for \underline{O}bservation-based \underline{M}arkov chai\underline{N} ref\underline{I}nement (\acronym) and accurate QoS analysis of component-based systems. 

\revised{The \acronym\ method comprises the five activities shown in Fig.~\ref{fig:OMNI-workflow}.} The key characteristic of \acronym\ is its use of observed execution times for the components of the analysed system to refine a high-level abstract CTMC whose states correspond to the operations executed by these components. \revised{As such, the first \acronym\ activity is the collection of these} execution time observations\revised{, which} can come from unit testing the components prior to system integration, from logs of other systems that use the same components, or from the log of the analysed system. \revised{The second \acronym\ activity involves the development of a high-level CTMC model of the system under analysis. This model} can be generated from more general software models such as annotated UML activity diagrams as in \cite{GallottiGMT2008,6693145}, or can be provided by the system developers. \revised{The next \acronym\ activity requires the formalisation of the QoS properties of interest as continuous stochastic logic formulae.}

\begin{figure*}
\centering
\includegraphics[width=\linewidth]{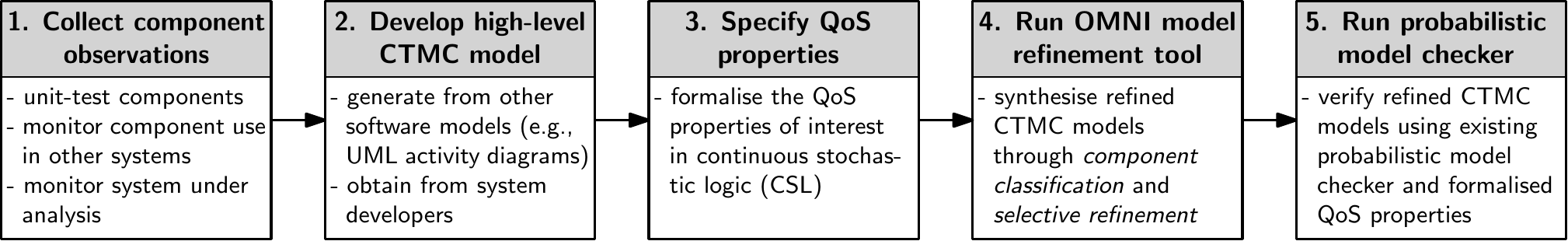}
\caption{\revised{\acronym\ workflow for the QoS analysis of component-based systems}}
\label{fig:OMNI-workflow}
\end{figure*}

\revised{The fourth activity of our \acronym\ method is the refinement of the high-level model.} \acronym\ avoids the synthesis of unnecessarily large and inefficient-to-analyse models by generating a different refined CTMC for each QoS property of interest. This generation of property-specific CTMCs is fully automated and comprises two steps. The first step, called \emph{component classification}, determines the effect of every system component on the analysed QoS property. The second step, called \emph{selective refinement}, produces the property-specific CTMC by using phase-type distributions~\cite{buchholz2014phase} to refine only those parts of the high-level CTMC that correspond to components which influence the QoS property being analysed. 
As such, \acronym-refined CTMCs model component executions with much greater accuracy than traditional CTMC modelling, whose exponential distributions match only the first moment of the unknown distributions of the observed execution times.

\revised{In the last activity of our method, t}he refined CTMC models generated by \acronym\ are analysed with the established probabilistic model checker PRISM~\cite{kwiatkowska2011prism}. As illustrated by the two case studies presented in the paper, these models support the accurate and efficient analysis of a broad spectrum of QoS properties specified in continuous stochastic logic \cite{aziz1996}. As such, \acronym's observation-enhanced QoS analysis can prevent many invalid engineering and business decisions associated with traditional CTMC-based QoS analysis. 

\revised{The \acronym\ activities 2, 3 and 5 correspond to the traditional method for QoS property analysis through probabilistic model checking. Detailed descriptions of these activities are available (e.g., in \cite{DBLP:journals/ife/Kwiatkowska13,DBLP:journals/mscs/BaierHHHK13,DBLP:journals/tse/FilieriTG16,DBLP:journals/jss/FrancoCBRSG16,DBLP:journals/tr/CalinescuGJPRT16}) and therefore we do not focus on them in this paper. Activities 1 and 4 are specific to \acronym. However, the tasks from activity~1 are standard software engineering practices, so the focus of our paper is on the observation-based refinement techniques used in the fourth activity of the \acronym\ workflow.}

\revised{Like most methods for software performance engineering \cite{10.1007/BFb0013866,smith2007introduction,Woodside:2007:FSP:1253532.1254717}, \acronym\ supports both the design of new systems and the verification of existing systems. Using \acronym\ to assess whether a system under design meets its QoS requirements or to decide a feasible service-level agreement for a system being developed requires the collection of component observations by unit testing the intended system components, or by monitoring other systems that use these  components. In contrast, for the verification of the QoS properties of an existing system, component observations can be collected using any of the techniques listed under the first activity from Fig.~\ref{fig:OMNI-workflow}, or a combination thereof.}

A preliminary version of \acronym\ that did not include the component classification step was introduced in \cite{OMNIref}. This paper extends the theoretical foundation from \cite{OMNIref} with key results that enable component classification, and therefore the synthesis of much smaller and faster to analyse refined CTMCs than those generated by our preliminary \acronym\ version. This extension is presented in Section~\ref{sect:comp-class},  implemented by our new \acronym\ tool described in Section~\ref{section:tool}\revised{, and shown to reduce  verification times by 54--74\% (compared to the preliminary \acronym\ version) in Section~\ref{subsect:rq4}. This is a particularly significant improvement because the same QoS property is often  verified many times, to identify suitable values for the parameters of the modelled system (e.g., see the case studies from \cite{Kwiatkowska2007:SFM,NP2014,GCT2018,Calinescu-etal-2017-SA,CCGMP2018}).} Additionally, we considerably extended and improved the validation of \acronym\ by evaluating it \revised{for the following two systems:}
\squishlist
\item[1)] A service-based system \revised{that we implemented using six real-world} web services \revised{--- two commercial web services provided by Thales Group, three free Bing web services provided by Microsoft, and a free WebserviceX.Net web service. The evaluation of \acronym\ for this system was based on lab experiments.}
\item[2)] The IT support system \revised{at the Federal Institute of Education, Science and Technology of Rio Grande do Norte (IFRN), Brazil. This system has over 44,000 users --- students and IFRN employees (including IT support staff). The evaluation of \acronym\ for this system was based on real datasets obtained from the system logs.}
\squishend

The rest of the paper is structured as follows. Section~\ref{section:preliminaries} introduces the notation, terminology and theoretical background for our work. Section~\ref{section:motivation} describes the service-based system used to evaluate \acronym, as well as to motivate and illustrate our QoS analysis method throughout the paper. The assumptions and theoretical results underlying the component classification and selective refinement steps of \acronym\ are presented in Section~\ref{section:refinement}, and the tool that automates their application is described in Section~\ref{section:tool}. Section~\ref{section:evaluation} evaluates the effectiveness of \acronym\ for the two systems mentioned above. This evaluation shows that, compared to traditional CTMC-based QoS analysis, our method (a)~reduces analysis errors by 84.4--89.6\% for the service-based system and by 94.7--97\% for the IT support system; and (b)~lowers the risk of invalid engineering and business decisions. The experimental results also show a decrease of up to 71.4\% in QoS analysis time compared to our preliminary \acronym\ version from \cite{OMNIref}. \revised{Section~\ref{sect:threats} discusses the threats to the validity of our results.} The paper concludes with an overview of related work in Section~\ref{section:related} and a brief summary in Section~\ref{section:conclusion}.


\section{Preliminaries}
\label{section:preliminaries}
\subsection{Continuous-time Markov chains \label{subsect:ctmcs}}

Continuous-time Markov chains~\cite{norris1998} are mathematical models for continuous-time stochastic processes over countable state spaces. To support the presentation of \acronym, we will use the following formal definition adapted from \cite{Kwiatkowska2007:SFM,Baier-etal-2003}.

\begin{definition}
\label{def:CTMC}
A continuous-time Markov chain (CTMC) is a tuple
\begin{equation}
  \mathcal{M}= (S,\bm{\pi},\mathbf{R}),
  \label{eq:CTMC} 
\end{equation}
where $S$ is a finite set of states, $\bm{\pi}:S\to [0,1]$ is an initial-state probability vector such that the probability that the CTMC is initially in state $s_i\!\in\! S$ is given by $\bm{\pi}(s_i)$ and $\sum_{s_i\in S} \bm{\pi}(s_i)\!=\!1$, and $\mathbf{R}\!:\! S\!\times\! S\!\to\! \mathbb{R}$ is a transition rate matrix such that, for any states $s_i\!\neq\! s_j$ from $S$, $\mathbf{R}(s_i,s_j)\!\geq\! 0$ specifies the rate with which the CTMC transitions from state $s_i$ to state $s_j$, and $\mathbf{R}(s_i,s_i)\! =\! -\sum_{s_j\in S\setminus \{s_i\}}\! \mathbf{R}(s_i,s_j)$.
\end{definition}

\noindent We will use the notation $\mathsf{CTMC}(S,\bm{\pi},\mathbf{R})$ for the continuous-time Markov chain $\mathcal{M}$ from~(\ref{eq:CTMC}). The probability that this CTMC will transition from state $s_i$ to another state within $t$ time units is 
\[
  1\!-\!e^{-t\cdot \sum_{s_k\!\in\! S\setminus\! \{s_i\}} \mathbf{R}(s_i,s_k)}
\]
and the probability that the new state is $s_j\!\in\! S\setminus\{s_i\}$ is
\begin{equation}
\label{eq:pij}
  p_{ij}=\mathbf{R}(s_i,s_j)\; /\; \sum_{s_k\in S\setminus\{s_i\}}\;\! \mathbf{R}(s_i,s_k).
\end{equation}
A state $s$ is an \emph{absorbing state} if $\textbf{R}(s, s')=0$ for all $s' \in S$, and a \emph{transient state} otherwise.

The properties of a CTMC $\mathcal{M}$ are analysed over its set of finite and infinite paths $\mathit{Paths}^\mathcal{M}$. A \emph{finite path} is a sequence $s_1t_1s_2t_2\ldots s_{k-1}t_{k-1}s_k$, where $s_1,s_2,$ $\ldots,s_k\in S$, $\bm{\pi}(s_1)\!>\!0$, $s_k$ is an absorbing state, and, for all $i\!=\!1,2,\ldots, k-1$, $\mathbf{R}(s_i,s_{i+1})\!>\!0$ and $t_i\!>\!0$ is the time spent in state $s_i$. An \emph{infinite path} from $\mathit{Paths}^\mathcal{M}$ is an infinite sequence $s_1t_1s_2t_2\ldots$ where $\bm{\pi}(s_1)\!>\!0$, and, for all $i\!\geq\! 1$, $s_i\!\in\! S$, $\mathrm{R}(s_i,s_{i+1})\!>\!0$ and the time spent in state $s_i$ is $t_i\!>\!0$. For any path $\omega\!\in\!\mathit{Paths}^\mathcal{M}$, the state occupied by the path at time $t\!\geq\! 0$ is denoted $\omega@t$. For infinite paths, $\omega@t\!=\!s_i$, where $i$ is the smallest index for which $t\!\leq\! \sum_{j=1}^i t_j$. For finite paths, $\omega@t$ is defined similarly if $t\leq \sum_{j=1}^{k-1} t_j$, and $\omega@t=s_k$ otherwise. Finally, the $i$-th state on the path $\omega$ is denoted $\omega[i]$, where $i\in\mathbb{N}_{>0}$ for infinite paths and $i\in\{1,2,\ldots,k\}$ for finite paths.

Continuous-time Markov chains are widely used for the modelling and analysis of stochastic systems and processes from domains as diverse as engineering, biology and economics \cite{anderson2012continuous,yin2012continuous}. In this paper, we focus on the use of CTMCs for the modelling and QoS analysis of component-based software and IT systems. These systems are increasingly important for numerous practical applications, and advanced probabilistic model checkers such as PRISM \cite{kwiatkowska2011prism}, MRMC \cite{katoen2011ins} and Storm \cite{DJKV2017} are available for the efficient analysis of their CTMC models.

\subsection{Continuous stochastic logic \label{subsect:CSL}}

CTMCs support the analysis of QoS properties expressed in \emph{continuous stochastic logic} (CSL) \cite{aziz1996}, which is a temporal logic with the syntax defined below.

\begin{definition}
Let AP be a set of atomic propositions, $a\!\in\!AP \mbox{, } p\! \in\! [0,1] \mbox{, } I$  an interval in  $\mathbb{R} \mbox{ and } \bowtie\!\! \mbox{}  \in \{\geq, >, <, \leq \}$. Then a state formula $\Phi$ and a path formula $\Psi$ in continuous stochastic logic are defined by the following grammar:
\begin{equation}
\begin{array}{l}
\Phi ::= true\, |\,  a\, |\, \Phi \wedge \Phi\, |\, \neg\Phi\, |\, P_{\bowtie p}[\Psi]\, |\, \mathcal{S}_{\bowtie p}[\Phi]\\
\Psi ::= X \Phi\, |\, \Phi U^{I} \Phi  
\end{array}.
\end{equation}
\end{definition}

\noindent
CSL formulae are interpreted over a CTMC whose states are \emph{labelled} with atomic propositions from $AP$ by a function $L\! :\! S\! \to\! 2^{AP}$. The (\emph{transient-state}) \emph{probabilistic operator} $P$ and the \emph{steady-state operator} $\mathcal{S}$ define bounds on the probability of system evolution. \emph{Next path formulae} $X \Phi$ and \emph{until path formulae} $\Phi_1 U^{I} \Phi_2$ can occur only inside the probabilistic operator $P$. 

The semantics of CSL is defined with a satisfaction relation $\models$ over the states $s\!\in\! S$ and the paths $\omega\!\in\! \mathit{Paths}^\mathcal{M}$ of a CTMC \cite{Baier-etal-2003}. \acronym\ improves the analysis of QoS properties expressed in the transient fragment of CSL,\footnote{Steady-state properties only depend on the first moment of the distributions of the times spent in the CTMC states, so they are already computed accurately by existing CTMC analysis techniques.} with semantics defined recursively by:
\[
\!
\begin{array}{ll}
s\models \mathit{true} & \!\forall s\in S\\
s\models a & \!\textrm{iff } \!a\in L(s)\\
s\models \Phi_1\wedge \Phi_2 & \!\textrm{iff } s\models \Phi_1 \wedge s\models \Phi_2 \\
s\models \neg \Phi & \!\textrm{iff } \neg(s\models \Phi) \\
s\models P_{\bowtie p}[\Psi] & \!\textrm{iff } \mathit{Pr}_s \{\omega\in\mathit{Paths}^\mathcal{M}\mid \omega\models \Psi\}\bowtie p \\
\omega \models X \Phi & \!\textrm{iff } \omega=s_1t_1s_2\ldots \wedge s_2\models \Phi\\
\omega \models \Phi_1 U^{I} \Phi_2 & \!\textrm{iff } \exists t\!\in\! I. (\forall t'\!\in\![0,t).\, \omega@t'\!\models\! \Phi_1)\wedge \omega@t\!\models\! \Phi_2
\end{array}
\]
where a formal definition for the probability measure $\mathit{Pr}_s$ on paths starting in state $s$ is available in \cite{Kwiatkowska2007:SFM,Baier-etal-2003}. \revised{Note how according to these semantics \cite{Baier-etal-2003}, until path formulae $\Phi_1 U^{I} \Phi_2$ are satisfied by a path $\omega$ if and only if $\Phi_2$ is satisfied at some time instant $t$ in the interval $I$ and $\Phi_1$ holds at all previous time instants $t'$, i.e., for all $t'\in[0,t)$.}
Finally, a state $s$ satisfies a steady-state formula $\mathcal{S}_{\bowtie p}[\Phi]$ iff, having started in state $s$, the probability of the CTMC being in a state where $\Phi$ holds \emph{in the long run} satisfies the bound `$\bowtie p$'. 

The shorthand notation  $\Phi_1 U \Phi_2\equiv \Phi_1 U^{[0,\infty)} \Phi_2$ and $F^I \Phi \!\equiv\! \textsf{true}\;\! U^I  \Phi$ is used when $I\! =\! [0,\infty)$ in an until formula and when the first part of an until formula is $\textsf{true}$, respectively. Probabilistic model checkers also support CSL formulae in which the bound `$\bowtie p$' from $P_{\bowtie p}[\Psi]$ is replaced with `$=?$', to indicate that the computation of the actual bound is required. We distinguish between the probability $\mathit{Pr}_s \{\omega\in\mathit{Paths}^\mathcal{M}\mid \omega\models \Psi\}$ that $\Psi$ is satisfied by the paths starting in a state $s$, and the probability 
\[
\begin{array}{l}
  P_{=?}[\Psi]=\sum_{s\in S} \bm{\pi}(s)\mathit{Pr}_s \{\omega\in\mathit{Paths}^\mathcal{M}\mid \omega\models \Psi\}\\\textsf{\hspace*{1.06cm}}=\mathit{Pr}_{\bm{\pi}} \{\omega\in\mathit{Paths}^\mathcal{M}\mid \omega\models 
\Psi\} 
\end{array}
\]
that $\Psi$ is satisfied by the CTMC. In the analysis of system-level QoS properties, we are interested in computing the latter probability.

\subsection{Phase-type distributions \label{sect:PHD}}

\acronym\ uses \emph{phase-type distributions} (PHDs) to refine the relevant elements of the analysed high-level abstract CTMC. PHDs model stochastic processes where the event of interest is the time to reach a specific state, and are widely used in the performance modelling of systems from domains ranging from call centres to healthcare~\cite{fackrell2009modelling,ishay2002fitting,marshall2009simulating}. PHDs support efficient numerical and analytical evaluation~\cite{buchholz2014phase}, and can approximate arbitrarily close any continuous distribution with a strictly positive density in $(0,\infty)$~\cite{doi:10.1080/15326349908807560}, although PHD fitting of distributions with deterministic delays requires extremely large numbers of states. 

A PHD is defined as the distribution of the time to absorption in a CTMC with one absorbing state~\cite{buchholz2014phase}. The $N\geq 1$ transient states of the CTMC are called the \emph{phases} of the PHD. With the possible reordering of states, the transition rate matrix of this CTMC can be expressed as:
\begin{equation}
\mathbf{R} = \left[ \begin{array}{c|c}
 \mathbf{D}_0 & \mathbf{d_1} \\ \hline
 \mathbf{0} & 0\\
 \end{array}
\right],
\end{equation}
where the $N\times N$ sub-matrix $\mathbf{D}_0$ specifies only transition rates between transient states, $\mathbf{0}$ is a $1\times N$ row vector of zeros, and $\mathbf{d}_1$ is an $N\times 1$ vector whose elements specify the transition rates from the transient states to the absorbing state. The elements from each row of $\mathbf{R}$ add up to zero (cf.\ Definition~1), so we additionally have $\mathbf{D}_0\mathbf{1}+\mathbf{d_1}=\mathbf{0}$, where $\mathbf{1}$ and $\mathbf{0}$ are column vectors of $N$ ones and $N$ zeros, respectively. Thus, $\mathbf{d_1}=-\mathbf{D}_0\mathbf{1}$ and the PHD associated with this CTMC is fully defined by the sub-matrix $\mathbf{D}_0$ and the row vector $\bm{\pi}_0$ containing the first $N$ elements of the initial probability vector $\bm{\pi}$ (as in most practical applications, we are only interested in PHDs that are acyclic and that cannot start in the absorbing state). We use the notation $\mathsf{PHD}(\bm{\pi}_0,\mathbf{D}_0)$ for this PHD.

\subsection{\revised{Erlang distributions} \label{sect:Erlang}}

\revised{The Erlang distribution~\cite{stewart2009probability} is a form of PHD in which $k$ exponential phases, each with the same rate parameter $\lambda$, are placed in series. The Erlang distribution has a $(k+1)$-element initial probability vector $\bm{\pi} \!=\! \begin{bmatrix}1 & 0 & \ldots & 0\end{bmatrix}$, such that the system always starts in an initial state $s_0$ and successively traverses states $s_1, s_2, \ldots$ until it reaches an absorbing state $s_k$. The distribution represents the expected time to reach the absorbing state, and has the cumulative distribution function
\begin{equation}
F(k,\lambda,x) = 1 - \sum_{i=0}^{k-1} \frac{(\lambda x)^{i}}{i!} e^{-\lambda x},
\end{equation}
for $x\geq 0$, the mean $m\!=\!\frac{k}{\lambda}$, and the variance $\frac{k}{\lambda^2}= \frac{m^2}{k}$ (which approaches zero as $k \rightarrow \infty$).}


\section{Motivating example: QoS analysis of a web application \label{section:motivation}}

\begin{figure}
\centering
  \includegraphics[width=6cm]{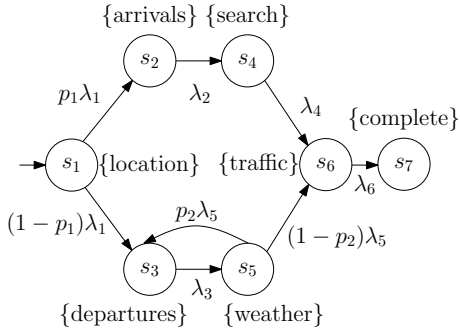}
  \caption{High-level abstract CTMC modelling the handling of a request by the web application\label{fig:CTMC}}
\end{figure}

To illustrate the limitations of traditional CTMC-based QoS analysis, we consider a travel web application that handles two types of requests:
\squishlist
\item[1.]  Requests from users who plan to meet and entertain a visitor arriving by train. 
\squishend

\squishlist
\item[2.] Requests from users looking for a possible destination for a day trip by train. 
\squishend
The handling of these requests by the application is modelled by the high-level abstract CTMC from Fig.~\ref{fig:CTMC}, which can be obtained from a UML activity diagram of the application.   \revised{The method for obtaining a Markov chain from an activity diagram is described in detail in \cite{GallottiGMT2008,6693145,da2017self,DBLP:conf/tase/CalinescuR13}. This method  requires annotating the outgoing edges of decision nodes from the diagram with the probabilities with which these edges are taken during the execution of the modelled application. Markov model states are then created for each of the activities, decision and start/end nodes in the diagram, and state transitions are added for each edge between these nodes; the transitions corresponding to outgoing edges of decision nodes ``inherit'' the probabilities that annotate these edges, while all other transition probabilities have a value of 1.0.} 

The initial state $s_1$ of the CTMC from Fig.~\ref{fig:CTMC} corresponds to finding the location of the train station. For the first request type, which is expected to occur with probability $p_1$, this is followed by finding the train arrival time (state $s_2$), identifying suitable restaurants in the area (state $s_4$),  obtaining a traffic report for the route from the user's location to the station (state $s_6$), and returning the response to the user (state $s_7$). 

For the second request type, which occurs with probability $1-p_1$, state $s_1$ is followed by finding a possible destination (state $s_3$), and obtaining a weather forecast for this destination (state $s_5$). With a probability of $p_2$ the weather is unsuitable and a new destination is selected  (back to state $s_3$). Once a suitable destination is selected, the traffic report is obtained for travel to the station (state $s_6$) and the response is returned to the user (state $s_7$).

\begin{table*}
\centering
\caption{Web services considered for the web application \label{tab:webservices}}
\sffamily
  \begin{footnotesize}
  \begin{tabular}{p{1.4cm}p{4.3cm}p{6.1cm}p{0.9cm}} 
  \toprule
  \textbf{Label} & \textbf{Thid-party service} & \textbf{URL} & \mbox{\textbf{rate} $\!(\!s^{\textrm{-}1}\!)$}\\ 
  \midrule
  location & Bing location service & \url{http://dev.virtualearth.net/REST/v1/Locations} & 9.62\\
  arrivals & Thales rail arrival board & \url{http://www.livedepartureboards.co.uk/ldbws/} &19.88\\
  departures & Thales rail departures board & \url{http://www.livedepartureboards.co.uk/ldbws/} &19.46\\
  search & Bing web search & \url{https://api.datamarket.azure.com/Bing/Search} &1.85\\
  weather & WebserviceX.net weather service & \url{http://www.webservicex.net/globalweather.asmx} &1.11\\
  traffic & Bing traffic service & \url{http://dev.virtualearth.net/REST/v1/Traffic} & 2.51\\ 
  \bottomrule
  \end{tabular} \rmfamily
  \end{footnotesize}
\end{table*}

\begin{figure*}
\centering
\includegraphics[width=0.3\linewidth]{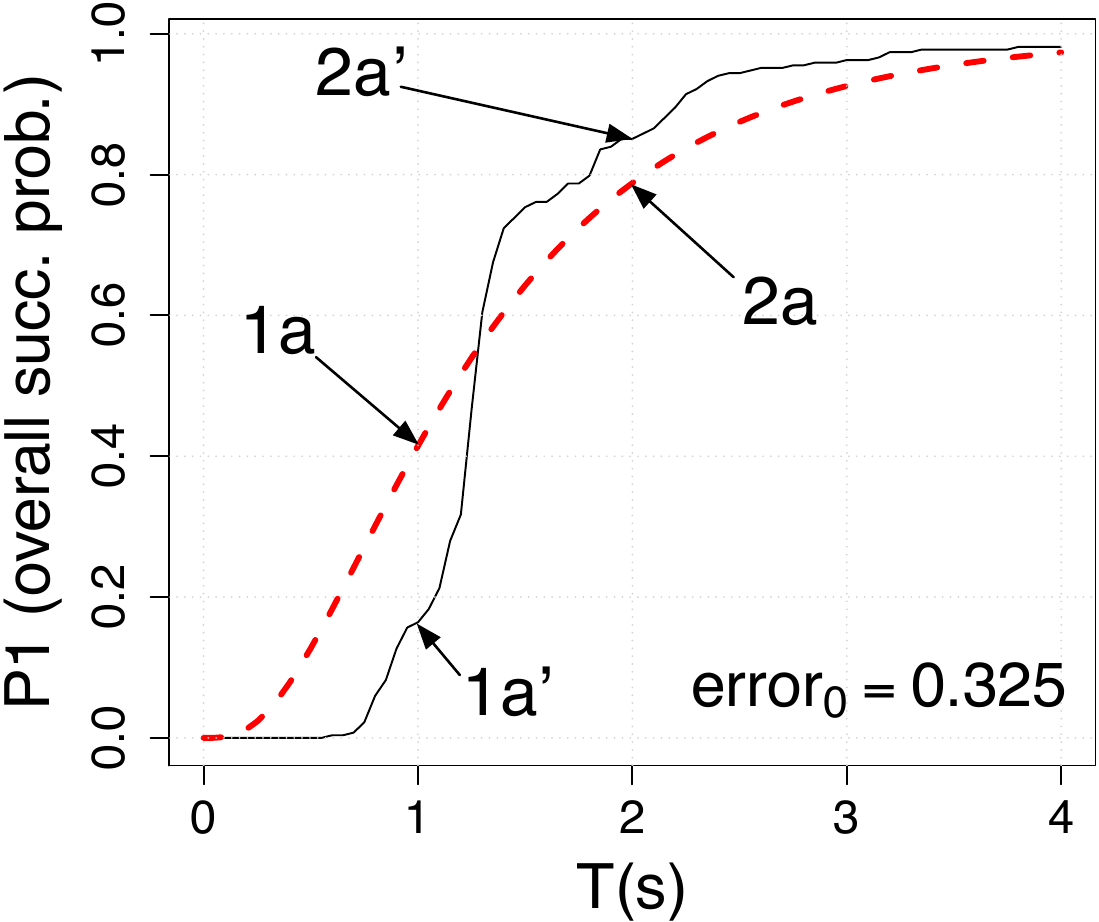}\hspace*{5mm}
\includegraphics[width=0.3\linewidth]{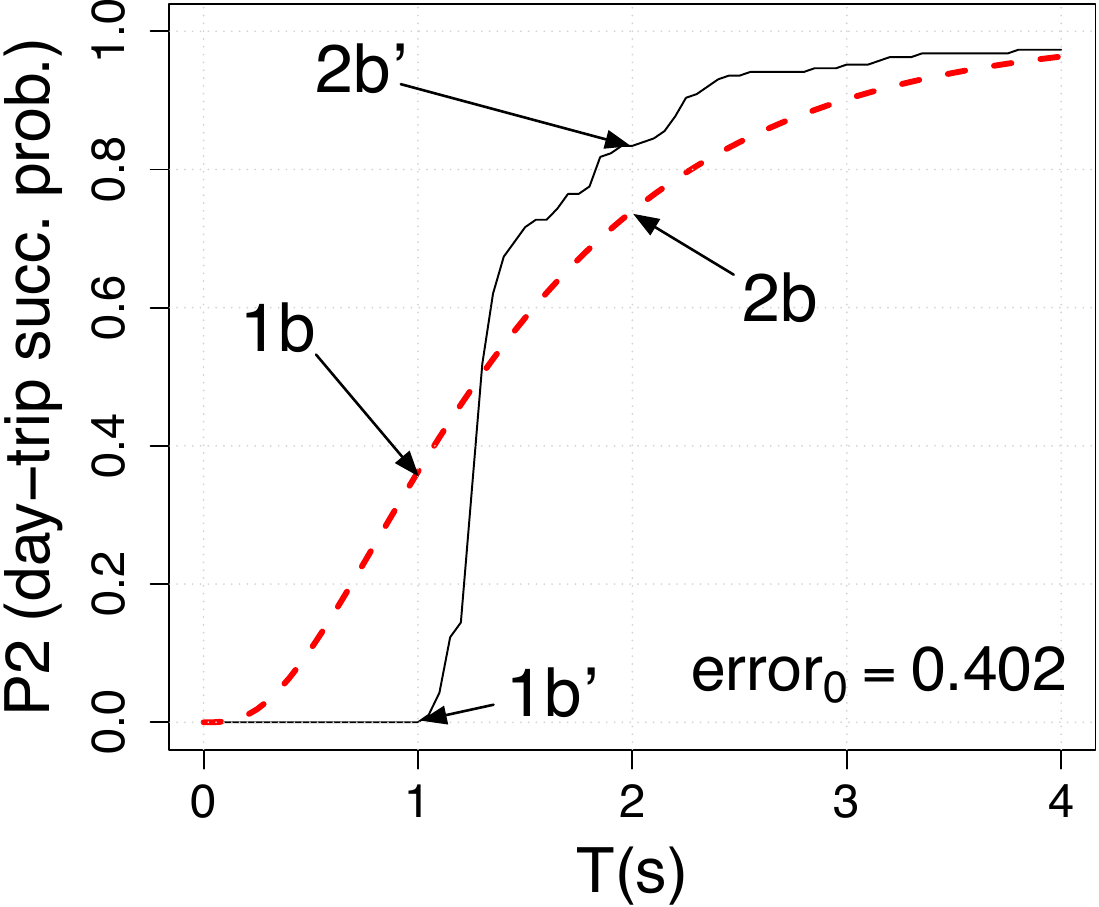}\hspace*{5mm}
\includegraphics[width=0.3\linewidth]{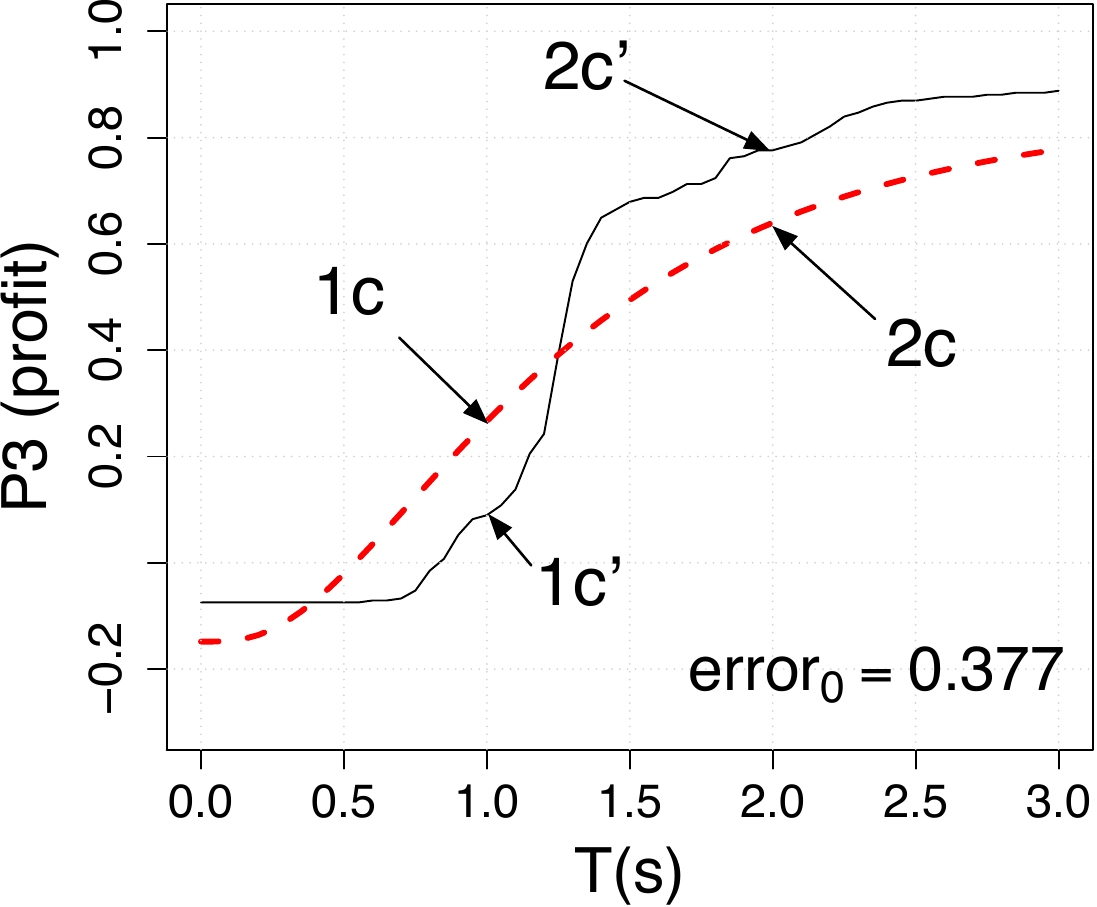}
\caption{Predicted (dashed lines) versus actual (continuous lines) property values}
\label{fig:VerExp}
\end{figure*}

The component execution rates $\lambda_1$ to $\lambda_6$ depend on the implementations used for these components, and we consider that a team of software engineers wants to decide if the real web services from Table~\ref{tab:webservices} are suitable for building the application. If they are suitable, the engineers need: 
\squishlist
\item[1.]  To select appropriate request-handling times to be specified in the application service-level agreement (SLA);
\item[2.]  To choose a pricing scheme for the application. 
\squishend
Accordingly, the engineers want to assess several QoS properties of the travel application variant built using these publicly available web services:\\

\noindent
\hspace*{-3mm}
\begin{tabular}{p{4mm}p{7.9cm}}
\textbf{P1} & The \revised{probability of successfully handling} user requests in under $T$ seconds, for $0\!<\!T\!\leq\! 4$.\\
\textbf{P2} &  The \revised{probability of successfully handling} ``day trip'' requests in under $T$ seconds, for $0\!<\!T\!\leq\! 4$.\\
\textbf{P3} &  The expected profit per request handled, assuming that 1~cent is charged for requests handled within $T$ seconds and a 2-cent penalty is paid for requests not handled within 3 seconds,  for $0<T\leq 3$.
\end{tabular}

\vspace*{2mm}
Service response times are assumed exponentially distributed in QoS analysis based on CTMC (as well as queueing network) models. Therefore, the engineers use observed service execution times $t_{i1},\dots,t_{in}$ for service $i$ to estimate the service rate $\lambda_i$ as

\begin{equation}
  \label{eq:rate}
  \lambda_i = \left(\frac{t_{i1}+t_{i2}+\cdots + t_{in}}{n}\right)^{-1}.
\end{equation}

\noindent
These execution times can be taken from existing logs (e.g.\ of other applications that use the same services) or can be obtained through testing the web services individually.
%
Finally, a probabilistic model checker is used to analyse properties \textbf{P1}--\textbf{P3} of the resulting CTMC. For this purpose, the three properties are first formalised as transient-state CSL formulae:
\begin{equation}
\label{eq:web-reqs}
\textrm{\hspace*{-0.3cm}}\begin{array}{ll}
\textbf{P1} & \,P_{=?}[F^{[0,T]} \mathit{complete}]\\
\textbf{P2} & \,P_{=?}[\neg\mathit{arrivals}\; U^{[0,T]} \mathit{complete}] / (1-p_1)\\
\textbf{P3} & \,P_{=?}[F^{[0,T]} \mathit{complete}] - 2\cdot P_{=?}[F^{(3,\infty)} \mathit{complete}] 
\end{array}
\end{equation}
\revised{The value of $T$ to be specified in the SLA is unknown a priori and hence we evaluate each property for a range of $T$ values} where $0\!<\!T\!\leq\! 4$ for \textbf{P1} and \textbf{P2}, and $0\!<\!T\!\leq\! 3$ for \textbf{P3}.  

To replicate this process, we implemented a prototype version of the application and we used it to handle $270$ randomly generated requests for $p_1 \!=\! 0.3$ and $p_2\! =\! 0.1$. \revised{Obtaining transition probabilities for Markov chains from real-world systems, and the effects of transition probabilities on system performance, have previously been considered~\cite{ghezzi2014mining, su2013asymptotic}. To decouple these effects from those due to the temporal characteristics of component behaviours, we utilise fixed probabilities for our motivating example. However, for the second system used to evaluate \acronym\ (Section~\ref{subsect:second-case-study}) we extract the transition probabilities from system logs, showing that \acronym\ also provides significant improvements in verification accuracy in this setting.} We obtained sample execution times for each web service (between $81$ 
for \textsf{arrivals} and \textsf{search} and $270$ for \textsf{location} and \textsf{traffic}), and we applied~(\ref{eq:rate}) to these observations, calculating the estimate service rates from Table~\ref{tab:webservices}.
\revised{Note that these observations are equivalent to observations obtained from unit testing the six services separately. This is due to the statistical independence of the execution times of different services, which we confirmed by calculating the Pearson correlation coefficient of the observations for every pair of services -- the obtained coefficient values, between $-0.17$ and $+0.11$, indicate lack of correlation.}
We then used the model checker PRISM \cite{kwiatkowska2011prism} to analyse the CTMC for these rates, and thus to predict the values of properties~(\ref{eq:web-reqs}). 

To assess the accuracy of the predictions, we also calculated the actual values of these properties \revised{at each time value $T$} using detailed timing information logged by our application. \revised{The error associated with a single property evaluation  may be quantified as the absolute difference between actual and predicted values
	\begin{equation}
	\left| actual(T)-predicted(T) \right|
	\end{equation} 
}The predictions obtained through CTMC analysis and the actual property values \revised{across the range of $T$ values} are compared in Fig.~\ref{fig:VerExp}. The errors reported in the figure are calculated \revised{using the distance measure recommended for assessing the overall error of CTMC/PHD model fitting in ~\cite{bobbio1992ml, buchholz2014phase, reinecke2012cluster, horvath2000approximating}, i.e.,} the area difference between the actual and the predicted property values:
\begin{equation}
\label{eq:err}
  \mathit{error}=\int_0^{T_{\max}} \left| actual(T)-predicted(T) \right|\: dT,
\end{equation}  
where $T_\mathrm{max}\!=\! 4$ for properties \textbf{P1} and \textbf{P2}, and $T_\mathrm{max}\!=\!3$ for property \textbf{P3}.\footnote{\revised{Both underestimation and overestimation of QoS property values contribute to the error because both can lead to undesirable false positives or false negatives when assessing whether QoS requirements are met. For example, overestimates of the overall success probability of a system can falsely indicate that a requirement that places a lower bound on this probability is met and the system is safe to use (false negative), while underestimates of the same property can falsely indicate that the requirement is violated and the system should not be used (false positive).}} 
Later in the paper, we will use \revised{this error measure 
 }to assess the improvements in accuracy due to the \acronym\ model refinement. In this section we focus on the limitations of CTMC-based transient analysis. Therefore, recall that the software engineers must make their decisions based only on the predicted property values from Fig.~\ref{fig:VerExp}; two of these decisions \revised{and their associated scenarios are described below}.

\vspace*{2mm}
\noindent
\revised{\emph{Scenario 1.} The engineers note that: 
\squishlist
\item[(a)] the predicted overall success probability (property \textbf{P1}) at $T\!\!=\!\!1$s is $0.415$ (marked 1a in Fig.~\ref{fig:VerExp}), i.e., slightly over 40\% of the requests are predicted to be handled within 1s;
\item[(b)] the predicted day-trip success probability (property \textbf{P2}) at $T\!=\!1$s is $0.363$ (1b in Fig.~\ref{fig:VerExp}), i.e., over 36\% of the day-trip requests are predicted to be handled within 1s;
\item[(c)] the expected profit (property \textbf{P3}) at $T\!=\!1$s, i.e., when charging 1~cent for requests handled within 1s, is $0.27$~cents (1c in Fig.~\ref{fig:VerExp}).
\squishend
Accordingly, the engineers decide to u}se the services from Table~\ref{tab:webservices} to implement the travel web application, with an SLA ``promising''  that requests will be handled within 1s \revised{with 0.4 success probability}, ``day trip'' requests will be handled within 1s \revised{with 0.35 success probability}, and charging 1~cent for requests handled within 1s. As shown in Fig.~\ref{fig:VerExp}, \revised{the actual property values at $T=1$s are $0.164$ for \textbf{P1} (marked 1a$'$ in Fig.~\ref{fig:VerExp}), $0$ for \textbf{P2} (1b$'$ in Fig.~\ref{fig:VerExp}) and $0.09$~cents for \textbf{P3} (1c$'$ in Fig.~\ref{fig:VerExp})}, so this decision would be wrong -- both promises would be violated by a wide margin, and the actual profit would be under a third of the predicted profit.

\vspace*{2mm}
\noindent
\revised{\emph{Scenario 2.} The engineers observe that the success probabilities of} handling requests or ``day trip'' requests within 2s \revised{are below 0.8 -- the predicted values for properties \textbf{P1} and \textbf{P2} at $T=2$s are $0.79$ (2a in Fig.~\ref{fig:VerExp}) and $0.74$ (2b in Fig.~\ref{fig:VerExp}), respectively;} and/or that the expected profit is below 0.7 cents per request when charging 1~cent for each request handled within 2s \revised{(2c in Fig.~\ref{fig:VerExp})}. 
\revised{As such, they decide to} look for alternative services for the application. As shown \revised{by points 2a'--2c'} in Fig.~\ref{fig:VerExp}, all the constraints underpinning this decision are actually satisfied, so the decision would also be wrong.

\vspace*{2mm}
\noindent
We chose the times and constrains in the two hypothetical decisions to show how the current use of idealised CTMC models in QoS analysis \emph{may} yield invalid decisions. The fact that choosing different times and constrains could  produce valid decisions is not enough: engineering decisions are meant to be consistently valid, not down to chance. It is this major limitation of traditional CTMC-based QoS analysis that our CTMC refinement method addresses as described in the next section.


\begin{figure*}
	\centering
	\begin{subfigure}[b]{1.0\textwidth}
	\includegraphics[width=1.0\linewidth]{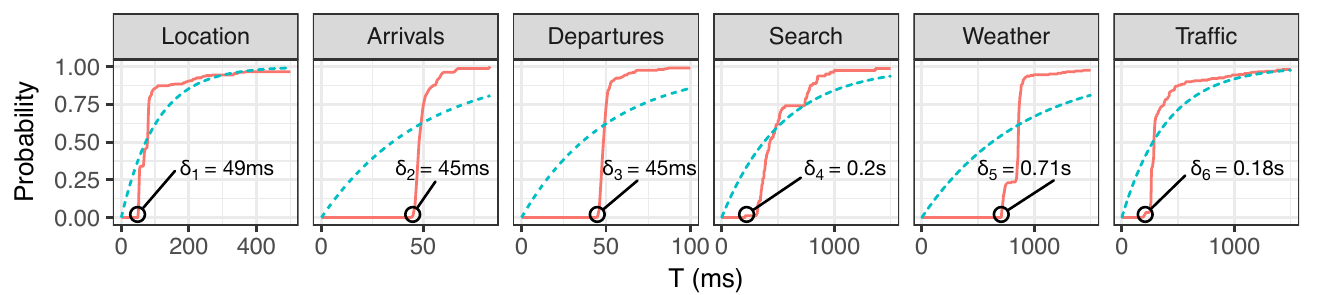}
	\caption{\revised{Empirical CDF for the service execution times (continuous lines) versus exponential models with rates computed from observed data (dashed lines)}}
	\end{subfigure}
	\begin{subfigure}[b]{1.0\textwidth}
		\includegraphics[width=\linewidth]{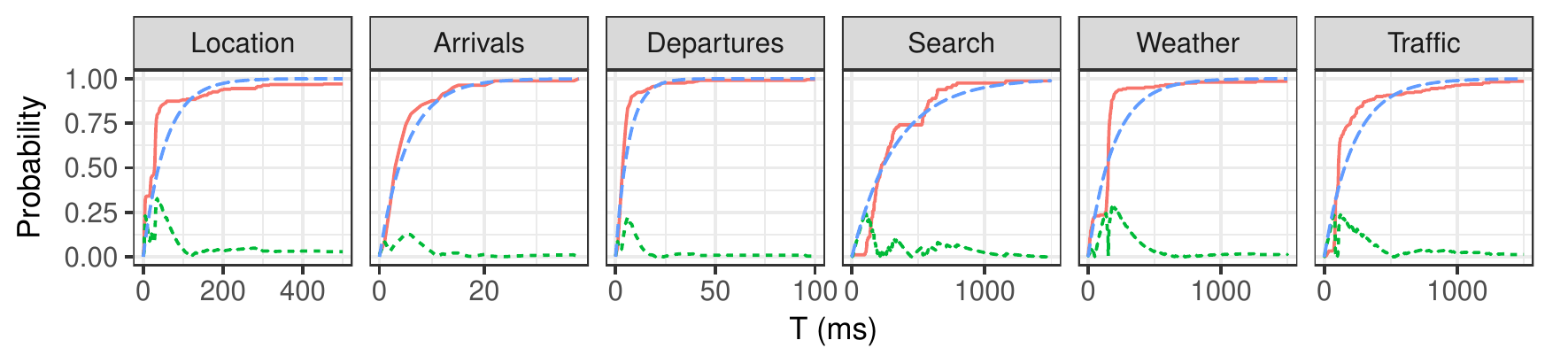}
		\caption{\revised{Empirical CDF for the service holding times (continuous lines) versus exponential models with rates computed from observed holding times (long dashed lines); for all services except Arrivals the difference between the two (short dashed lines) exceeds 20\% for multiple values of $T$}}
	\end{subfigure}
\caption{The services from the motivating example have non-zero delays and non-exponentially distributed holding times
	\label{fig:CDFDelay}}
\end{figure*}	

\section{The \acronym\ method for CTMC refinement\label{section:refinement}}

\subsection{Overview \label{subsect:overview}}

\acronym\ addresses the refinement of high-level CTMC models $\mathsf{CTMC}(S,\bm{\pi},\mathbf{R})$ of software systems that satisfy the following assumptions:
\squishlist
\item Each state $s_i\in S$ corresponds to a component of the system, and $\bm{\pi}(s_i)$ is the probability that $s_i$ is the initial component executed by the system;
\item For any distinct states from $s_i,s_j\!\in\! S$, the transition rate $\mathbf{R}(s_i,s_j)\!=\!p_{ij}\lambda_i$, where $p_{ij}$ represents the (known or estimated) probability~(\ref{eq:pij}) that component $i$ is followed by component $j$ and $\lambda_i$ is obtained by applying~(\ref{eq:rate}) to $n_i$ observed execution times $\tau_{i1}, \tau_{i2}, \dots, \tau_{in_i}$ of component $i$;
\item Each state $s_i\!\in\! S$ is labelled with the name of its corresponding component, which we will call ``component $i$'' for simplicity.
\squishend
This CTMC model makes the standard assumption that component execution times are exponentially distributed. However, this assumption is typically invalid for two reasons. First, each component $i$ has a \emph{delay} $\delta_i$ (i.e.\ minimum execution time) approximated by
\begin{equation}
\label{eq:delay}
  \delta_i \approx \min_{j=1}^{n_i} \tau_{ij}
\end{equation}
such that its probability of completion within $\delta(s_i)$ time units is zero. In contrast, modelling the execution time of the component as exponentially distributed with rate $\lambda_i$ yields a non-zero probability $1-e^{-\lambda_i\delta_i}$ of completion within $\delta_i$ time units. Second, even the \emph{holding times} 
\begin{equation}
\label{eq:holding-times}
  \tau'_{i1}=\tau_{i1}-\delta_i,\; \tau'_{i2}=\tau_{i2}-\delta_i,\; \ldots,\; \tau'_{in_i}=\tau_{in_i}-\delta_i
\end{equation}  
of the component are rarely exponentially distributed. 

\begin{example}
\label{ex0}
\revised{Fig.~\ref{fig:CDFDelay}a shows the empirical cumulative distribution functions (CDFs) for the execution times of the six services from our motivating example (cf.\ Table~\ref{tab:webservices}), and the associated exponential models with rates given by (\ref{eq:rate}). The six services have minimum observed execution times $\delta_1$ to $\delta_6$ between 45ms and 0.71s (due to network latency and request processing time), and their exponential model is a poor representation of the observed temporal behaviour. Furthermore, the best-fit exponential model of the observed holding times for these services (shown in Fig.~\ref{fig:CDFDelay}b) is also inaccurate.}
\end{example}

\acronym\ overcomes these significant problems 
\revised{by generating} a refined CTMC for each QoS property of interest \revised{in two steps}, and uses standard probabilistic model checking to analyse the refined CTMC.
%
\revised{As shown in Fig.~\ref{fig:BlockDiagram}, 
the first \acronym\ step,} called \emph{component classification}, partitions the states of the high-level CTMC into subsets that require different types of refinement because of the different impact of their associated system components on the analysed property. For instance, components unused on an execution path have no effect on QoS properties (e.g.\ response time) associated solely with that path, and therefore their corresponding states from the high-level CTMC need not be refined. 
The second \acronym\ step, called \emph{selective refinement}, replaces the states which correspond to components that impact the analysed property with new states and transitions that model the delays and holding times of these components by means of Erlang distributions~\cite{stewart2009probability} and phase-type distributions (PHDs)~\cite{buchholz2014phase}, respectively. 

\begin{figure}
  \centering
  \includegraphics[width=\hsize]{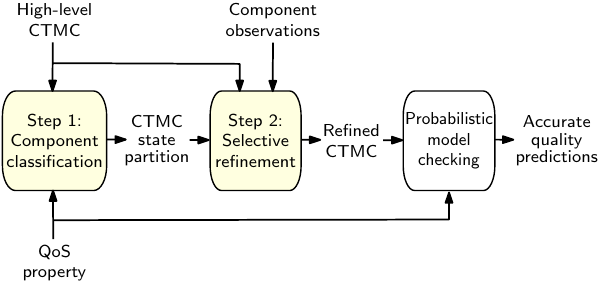}
  \caption{\acronym\ CTMC refinement and verification}
  \label{fig:BlockDiagram}
\end{figure}

As shown by our experimental results from Section~\ref{section:evaluation}, the two-step \acronym\ process produces refined CTMCs that are often much smaller and faster to analyse than the CTMCs obtained by obliviously refining every state of the high-level CTMC, e.g.\ as done in our preliminary work from \cite{OMNIref}. These benefits dominate the slight disadvantage of having to refine the high-level CTMC for each analysed property, which is further mitigated by our \acronym\ tool by caching and reusing refinement results across successive refinements of the same high-level CTMC, \revised{as described in Section~\ref{section:tool}}. Likewise, modelling the delay and holding time of system components separately (rather than using single-PHD fitting) yields smaller and more accurate refined models, in line with existing theory~\cite{doi:10.1080/15326349908807560} and our preliminary results from \cite{OMNIref}.

\revised{Several factors can impede or impact the success of our \acronym\ method:
\squishlist
\item[1.] Components with execution times that are not statistically independent. Markov models assume that the transition rates associated with different states are statistically independent. If the execution times of different components are not independent (e.g., because the components are running on the same server), then this premise is not satisfied, and \acronym\ cannot be applied.
\item[2.]  Changing component behaviour. If the system components change their behaviour significantly over time, then \acronym\ cannot predict the changed behaviour. This is a more general difficulty with model-based prediction.
\item[3.] Insufficient observations of component execution times. The accuracy of \acronym-refined models decreases when fewer observations of the system components are available. We provide details about the impact of the training dataset size on the \acronym\ accuracy in Section~\ref{subsubsect:rq3}.
\squishend
}

The component classification and selective refinement steps of \acronym\ are presented in the rest of this section.

\subsection{Component classification \label{sect:comp-class}}

Given a high-level CTMC model $\mathsf{CTMC}(S,\bm{\pi},\mathbf{R})$ of a system, and a QoS property encoded by the transient CSL formula $P_{=?}[\Phi_1 U^{I} \Phi_2]$, this \acronym\ step builds a partition 
\begin{equation}
  S=S_\textsf{X}\cup S_\textsf{O} \cup S_1\cup S_2\cup \ldots \cup S_m
  \label{eq:Spartition}
\end{equation}
of the state set $S$. 
Intuitively, the ``eXclude-from-refinement'' set $S_\textsf{X}$ will contain states with zero probability of occurring on paths that satisfy $\Phi_1 U^{I} \Phi_2$; the ``Once-only'' set $S_\textsf{O}$ will contain states with probability $1.0$ of appearing once and only once on every path that satisfies $\Phi_1 U^{I} \Phi_2$; and each ``together'' set $S_i$ will contain states that can only appear as a sequence on paths that satisfy $\Phi_1 U^{I} \Phi_2$. \revised{Formal 
definitions of} the disjoint sets $S_\textsf{X}$, $S_\textsf{O}$, and $S_1$ to $S_m$ \revised{and descriptions of their roles in \acronym\ are provided in Sections~\ref{subsubsec:exclude}--\ref{subsubsec:together}.} 

\revised{\subsubsection{Exclude-from-refinement state sets \label{subsubsec:exclude}}}

\begin{definition}
\revised{The \emph{exclude-from-refinement state set} $S_\textsf{X}$ associated with an until path formula $P_{=?}[\Phi_1 U^{I} \Phi_2]$} over the continuous-time Markov chain $\mathcal{M}=\mathsf{CTMC}(S,\bm{\pi},\mathbf{R})$ is the set of CTMC states
\begin{equation}
  S_\textsf{X}=\{s\in S \mid P_{=?}[(\neg s \wedge \Phi_1) U \Phi_2] = P_{=?}[\Phi_1 U \Phi_2] \},
  \label{eq:Sx}
\end{equation}
where, for each state $s\in S$, $AP$ is extended with an atomic proposition also named `$s$' that is true in state $s$ and false in every other state. \revised{Thus, $S_\textsf{X}$ comprises all states $s$ for which the probability $P_{=?}[(\neg s \wedge \Phi_1) U \Phi_2]$ of reaching a state satisfying $\Phi_2$ along paths that \emph{do not contain state $s$} and on which $\Phi_1$ holds in all preceding states is the same as the probability $P_{=?}[\Phi_1 U \Phi_2]$ of reaching a state that satisfies $\Phi_2$ along paths on which $\Phi_1$ holds in all preceding states.}
\end{definition}

\begin{theorem}
\label{th:Sx}
\revised{Let $S_\textsf{X}$ be the exclude-from-refinement state set associated with the until path formula $P_{=?}[\Phi_1 U^{I} \Phi_2]$} over the continuous-time Markov chain $\mathcal{M}=\mathsf{CTMC}(S,\bm{\pi},\mathbf{R})$ with atomic proposition set $\mathit{AP}$. 
Then, for any $I\!\subseteq\! \mathbb{R}_{\geq 0}$, the probability $P_{=?}[\Phi_1 U^{I}\Phi_2]$ does not depend on the transition times from states in $S_\textsf{X}$.
\end{theorem}

\begin{proof}
The proof is by contradiction. Consider a generic state $s_\textsf{X}\in S_\textsf{X}$ and the following sets of paths:
\begin{equation*}
\begin{array}{l}
\!\!A\!=\!\{\omega\!\in\!\mathit{Paths}^\mathcal{M}\!\mid\! \exists t>0\,.\,(\omega@t\!\models\! \Phi_2\,\wedge\\ 
    \textsf{\hspace*{4.8cm}}(\forall t'\!\in\![0,t). \omega@t'\!\models\!\Phi_1))\} \\
\!\!B\!=\!\{\omega\!\in\!\mathit{Paths}^\mathcal{M}\!\mid\! \exists t>0 \,.\,(\omega@t\!\models\! \Phi_2 \,\wedge\\
    \textsf{\hspace*{2.9cm}}(\forall t'\!\in\![0,t). \omega@t'\!\models\!\Phi_1\wedge \omega@t'\!\neq\! s_\textit{X}))\} \\
\!\!C\!=\!\{\omega\!\in\!\mathit{Paths}^\mathcal{M}\!\mid\! \exists t>0 \,.\,(\omega@t\!\models\! \Phi_2 \,\wedge\\ 
   \textsf{\hspace*{1cm}}(\forall t'\!\in\![0,t). \omega@t'\!\models\!\Phi_1\!)\wedge(\exists t'\in [0,t). \omega@t'=s_\textsf{X}))\}
\end{array}
\end{equation*}
As $A\!=\!B\cup C$ and $B\cap C\!=\!\emptyset$, we have $\mathit{Pr}_{\bm{\pi}}(A)=\mathit{Pr}_{\bm{\pi}}(B)+\mathit{Pr}_{\bm{\pi}}(C)$. However, according to~(\ref{eq:Sx}), $\mathit{Pr}_{\bm{\pi}}(A) \!=\! P_{=?}[\Phi_1 U \Phi_2]$ $=\! P_{=?}[(\neg s \wedge \Phi_1) U \Phi_2] \!=\! \mathit{Pr}_{\bm{\pi}}(B)$, so $\mathit{Pr}_{\bm{\pi}}(C)\!=\!0$.

Assume now that the time spent by the CTMC in state $s_\textsf{X}$ has an impact on the value of $P_{=?}[\Phi_1 U^{I} \Phi_2]$ over $\mathit{Paths}^\mathcal{M}$ for an interval $I\subseteq \mathbb{R}_{\geq 0}$. This requires that, at least for some (possibly very small) values of the time $t_\textit{X}>0$ spent in $s_\textsf{X}$, $s_\textsf{X}$ appears on paths from a set
\[
\begin{array}{l}
\!\!C'\!=\!\{\omega\!\in\!\mathit{Paths}^\mathcal{M}\!\mid\! \exists t\!\in\! I.\,(\omega@t\!\models\! \Phi_2 \,\wedge\\ 
   \textsf{\hspace*{1cm}}(\forall t'\!\in\![0,t). \omega@t'\!\models\!\Phi_1\!)\wedge(\exists t'\in [0,t). \omega@t'=s_\textsf{X}))\}
\end{array}
\]
such that $\mathit{Pr}_{\bm{\pi}}(C')>0$; otherwise, varying $t_\textit{X}$ cannot have any impact on
\[
  \begin{array}{l}
    \!P_{=?}[\Phi_1 U^{I} \Phi_2] = \mathit{Pr}_{\bm{\pi}}\{\omega\!\in\!\mathit{Paths}^\mathcal{M}\!\mid\! \exists t\in I\,.\, (\omega@t\!\models\! \Phi_2\,\wedge\\ 
    \textsf{\hspace*{5.1cm}}(\forall t'\!\in\![0,t). \omega@t'\!\models\!\Phi_1\!))\}
  \end{array}
\]
However, since $C'\!\subseteq\! C$ we must have $\mathit{Pr}_{\bm{\pi}}(C)\!\geq\!\mathit{Pr}_{\bm{\pi}}(C')\!$ $>\!0$, which contradicts our earlier finding that $\mathit{Pr}_{\bm{\pi}}(C)\!=\!0$, completing the proof.
\end{proof}

Theorem~\ref{th:Sx} allows \acronym\ to leave the states from $S_\textsf{X}$ unrefined with no loss of accuracy in the QoS analysis results. The theorem also provides a method for obtaining $S_\textsf{X}$ by computing the until formula $P_{=?}[(\neg s \wedge \Phi_1) U \Phi_2]$ for each state $s$ of the high-level CTMC (i.e.\ for each system component) and comparing the result with the value of the CSL formula $P_{=?}[\Phi_1 U \Phi_2]$, which is only computed once.  Existing probabilistic model checkers compute these \emph{unbounded} until formulae very efficiently, as they only depend on the probabilities~(\ref{eq:pij}) of transition between CTMC states and not on the state transition rates \cite{Kwiatkowska2007:SFM,Baier-etal-2003}.\revised{\footnote{\revised{To asses the time taken by model checking, an experiment was carried out to evaluate each state from the motivating example for inclusion in $S_\textsf{X}$. This experiment was repeated 30 times and the average time taken by model checking each state was found to be 1.6ms.}}}

\begin{example}
\label{ex1}
Consider the QoS properties~(\ref{eq:web-reqs}) of the web application from our motivating example.  For property \textbf{P2} and the high-level CTMC model from Fig.~\ref{fig:CTMC}, we have 
\[
\begin{array}{l}
  P_{=?}[\neg\mathit{arrivals}\; U\; \mathit{complete}]=1-p_1=\\
  \qquad P_{=?}[(\neg s_2\wedge\neg\mathit{arrivals})\; U\; \mathit{complete}]=\\
  \qquad P_{=?}[(\neg s_4\wedge\neg\mathit{arrivals})\; U\; \mathit{complete}] =\\
  \qquad P_{=?}[(\neg s_7\wedge\neg\mathit{arrivals})\; U\; \mathit{complete}],
\end{array}
\]
(and $P_{=?}[(\neg s\wedge\neg\mathit{arrivals})\, U \mathit{complete}] \!\neq\! 1\!-\!p_1$ for any other state $s$),
so $S_\textsf{X}\!=\!\{s_2,s_4,s_7\}$ for \textbf{P2}. Applying Theorem~\ref{th:Sx} to the other two properties from~(\ref{eq:web-reqs}) yields $S_\textsf{X}\!=\!\{s_7\}$.
\end{example}

\revised{\subsubsection{Once-only state sets \label{subsubsec:once}}}

\begin{definition}
\revised{The \emph{once-only state set} $S_\textsf{O}$ associated with an until path formula $P_{=?}[\Phi_1 U^{I} \Phi_2]$} over the continuous-time Markov chain $\mathcal{M}=\mathsf{CTMC}(S,\bm{\pi},\mathbf{R})$ is the set
\begin{equation}
\begin{array}{l}
  \!\!S_\textsf{O}\!=\!\{s\!\in\! S\!\setminus\! S_\textsf{X} \mid P_{>0}[\Phi_1 U \Phi_2] \wedge P_{\leq 0}[(\neg s\!\wedge\! \Phi_1) U \Phi_2] \:\wedge\\ 
  \;\qquad\quad\forall s'\in S\:.\: (s\models P_{>0}[X s'] \rightarrow s'\models P_{\leq 0}[\neg S_\textsf{X} U s])\},\\[1mm]
\end{array}
  \label{eq:So}
\end{equation}
where the until formula $\neg S_\textsf{X} U s$ holds for paths that reach state $s$ without going through any states from $S_\textsf{X}$ (which corresponds to labelling the states from $S_\textsf{X}$ with the atomic proposition `$S_\textsf{X}$').
\end{definition}

The next theorem 
\revised{asserts that for every state $s_\textsf{O}$ from $S_\textsf{O}$, $P_{=?}[\Phi_1 U^{I} \Phi_2]$ can be calculated by applying the probability measure $\mathit{Pr}_{\bm{\pi}}$ to the set of paths $\omega$ which, in addition to satisfying the clause specified by the CSL semantics (i.e., $\exists t\!\in\! I. (\forall t'\!\in\![0,t).\, \omega@t'\!\models\! \Phi_1)\wedge \omega@t\!\models\! \Phi_2$), contain $s_\textsf{O}$ once and only once \emph{before time instant $t$}. Using the unique existential quantifier $\exists!$, the last clause can be formalised as $\exists! i\:.\:(\omega[i]=s_\textsf{O} \wedge \sum_{j=1}^i t_j<t)$, where $t_j$ is the time spent in the $j$-th state on the path (cf.\ Section~\ref{subsect:ctmcs}).} 


\begin{theorem}
\label{th:So}
\revised{Let $S_\textsf{O}$ be the once-only state set associated with the until path formula $P_{=?}[\Phi_1 U^{I} \Phi_2]$} 
over the continuous-time Markov chain $\mathcal{M}=\mathsf{CTMC}(S,\bm{\pi},\mathbf{R})$. 
Then, for any state $s_\textsf{O}\in S_\textsf{O}$ and interval $I\!\subseteq\! \mathbb{R}_{\geq 0}$,
\begin{equation}
\begin{array}{l}
  \!\!\!\!P_{=?}[\Phi_1 U^{I} \Phi_2] = \mathit{Pr}_{\bm{\pi}} \{\omega\!\in\!\mathit{Paths}^\mathcal{M}\mid\\
    \!\quad \exists t\!\in\! I\:.\: (\omega@t\!\models\!\Phi_2\wedge (\forall t'\!\in\![0,t)\:.\: \omega@t'\!\models\!\Phi_1)\:\wedge\\
  \qquad\qquad\:\:\textrm{\revised{$\exists! i$}}\:.\: (\omega[i]=s_\textsf{O} \wedge \sum_{j=1}^i t_j<t))\}.
  \end{array}
  \label{eq:SOresult}
\end{equation}
\end{theorem}
\begin{proof}
Let $A'$ denote the subset of $\mathit{Paths}^\mathcal{M}$ from~(\ref{eq:SOresult}). According to CSL  semantics, $P_{=?}[\Phi_1 U^{I} \Phi_2]=\mathit{Pr}_{\bm{\pi}}(A)$ where
\[
  \begin{array}{l}
    A = \{\omega\!\in\!\mathit{Paths}^\mathcal{M}\!\mid\! \exists t\in I\,.\,(\omega@t\!\models\! \Phi_2\,\wedge\\ 
    \textsf{\hspace*{4.8cm}}(\forall t'\!\in\![0,t). \omega@t'\!\models\!\Phi_1\!))\}.
  \end{array}
\]
Since $A'\subseteq A$, we have $P_{=?}[\Phi_1 U^{I} \Phi_2]\!=\!\mathit{Pr}_{\bm{\pi}}(A)\!=\!\mathit{Pr}_{\bm{\pi}}(A')+\mathit{Pr}_{\bm{\pi}}(A\setminus A')$, so to prove the theorem we must show that $\mathit{Pr}_{\bm{\pi}}(A\setminus A')=0$. To this end, we partition $A\setminus A'$ into two disjoint subsets: $A_1$, comprising the paths that do not contain state $s_\textsf{O}$ before time $t$ from the first line of~(\ref{eq:SOresult}), and $A_2$, comprising the paths that contain state $s_\textsf{O}$ before time $t$ more than once. Since $P_{\leq 0}[(\neg s_\textsf{O}\!\wedge\! \Phi_1) U \Phi_2]$ holds (according to the definition of $S_\textsf{O}$), $\mathit{Pr}_{\bm{\pi}}(A_1)=0$. Similarly, since $\forall s'\in S\:.\: (s_\textsf{O}\models P_{>0}[X s'] \rightarrow s'\models P_{\leq 0}[\neg S_\textsf{X} U s_\textsf{O}])$ holds, the set of paths satisfying $\Phi_1 U^{I} \Phi_2$ and containing $s_\textsf{O}$ twice (without reaching states in $S_\textsf{X}$) occur with probability zero. As $A_2$ is included in this set, we necessarily have $\mathit{Pr}_{\bm{\pi}}(A_2)=0$. We conclude that $\mathit{Pr}_{\bm{\pi}}(A\setminus A')=\mathit{Pr}_{\bm{\pi}}(A_1)+\mathit{Pr}_{\bm{\pi}}(A_2)=0$, which completes the proof.
\end{proof}

\acronym\ exploits Theorem~\ref{th:So} in two ways. First, since $S_\textsf{O}$ states correspond to system components always executed before $\Phi_1 U^{I} \Phi_2$ becomes true, $P_{=?}[\Phi_1 U^{I} \Phi_2]=0$ for any interval $I\subseteq \left[0,\sum_{s_i\in S_\textsf{O}} \delta_i \right)$, where $\delta_i$ is the delay~(\ref{eq:delay}) of the component $i$ associated with state $s_i$. Therefore, \acronym\ returns a zero probability in this scenario without performing probabilistic model checking. Second, because the components associated with $S_\textsf{O}$ states are executed precisely once on relevant CTMC paths, no modelling of their delays is required, and \acronym\ only needs to model the holding times of these states. Importantly, obtaining $S_\textsf{O}$ to enable these simplifications only requires the probabilities of unbounded until and next path formulae (cf.~(\ref{eq:So})), which probabilistic model checkers can compute efficiently  for the reasons we explained earlier in this section.



\begin{example}
\label{ex2}
Consider property \textbf{P1} from the QoS properties~(\ref{eq:web-reqs}) in our motivating example: $P_{=?}[F^{[0,T]} \mathit{complete}]$. In line with definition~(\ref{eq:So}), we obtain the set $S_\textsf{O}$ for this property by first evaluating the following CSL formulae for the high-level CTMC from Fig.~\ref{fig:CTMC}:
\squishlist
\item $P_{>0}[\mathit{true}\:U\mathit{complete}]$ which holds as $P_{=?}[F\:\mathit{complete}]\!=\!1$
\item $P_{\leq 0}[(\neg s\wedge\mathit{true})\: U\mathit{complete}]\!=\!P_{\leq 0}[\neg s \, U\mathit{complete}]$, which holds only for states $s_1$ and $s_6$.
\squishend
The constraint \mbox{$\forall s'\!\in\! S. (s\!\models\! P_{>0}[X s'] \rightarrow s'\!\models\! P_{\leq 0}[\neg S_\textsf{X} U s])$} is then checked only for the $S_\textsf{O}$-candidate states $s=s_1$ and $s=s_6$, taking into account the fact that $S_\textsf{X}=\emptyset$ (cf.~Example~\ref{ex1}). For instance, since $s_1\models P_{>0}[X s']$ only for $s'\in \{s_2,s_3\}$, and $s_2,s_3\models P_{\leq 0}[\mathit{true}\: U\: s_1]$, we conclude that $s_1\in S_\textsf{O}$. Similarly, $s_6\models P_{>0}[X s']$ only if $s'=s_7$ and $s_7\models P_{\leq 0}[\mathit{true}\: U\: s_6]$, so $s_6\in S_\textsf{O}$, giving $S_\textsf{O}=\{s_1,s_6\}$. It is easy to show that the same ``once-only'' state set is obtained for the other two properties from~(\ref{eq:web-reqs}).
\end{example}

\revised{\subsubsection{Together state sets \label{subsubsec:together}}}

\begin{algorithm}[t]
  \caption{Generation of ``together'' state sequences \label{alg:together}}
  \begin{algorithmic}[1]
    \Function{TogetherSeqs}{$\mathsf{CTMC}(S,\bm{\pi},\mathbf{R})$, $S_\textsf{X}$, $S_\textsf{O}$}
      \State $\mathit{TS} \gets \emptyset$, $\mathit{States} \gets S\setminus (S_\textsf{X}\cup\ S_\textsf{O})$
      \While {$\mathit{States}\neq\emptyset$}
         \State $s\gets \textsc{PickAnyElement}(States)$
         \State $T \gets \langle s \rangle$, $\mathit{States}\gets \mathit{States}\setminus \{s\}$
         \State $\mathit{left}, \mathit{right}\gets \mathit{true}$
         \While {$(\mathit{left}\vee \mathit{right}) \wedge \mathit{States}\neq\emptyset$}
           \If{$\mathit{left}$}
              \State $s\gets\textsc{Pred}(\textsc{Head}(T),\mathit{States},S,\bm{\pi},\mathbf{R})$
              \If{$s\neq$ NIL}
                \State $T\gets \langle s\rangle^\frown T$, $\mathit{States}\gets \mathit{States}\setminus \{s\}$
              \Else
                \State $\mathit{left}\gets\mathit{false}$
              \EndIf
           \EndIf
           \If{$\mathit{right}$}
              \State $s\gets\textsc{Succ}(\textsc{Tail}(T),\mathit{States},S,\bm{\pi},\mathbf{R})$
              \If{$s\neq$ NIL}
                \State $T\gets T^\frown \langle s\rangle$, $\mathit{States}\gets \mathit{States}\setminus \{s\}$
              \Else
                \State $\mathit{right}\gets\mathit{false}$
              \EndIf
           \EndIf
         \EndWhile
         \State $\mathit{TS}\gets \mathit{TS} \cup \{T\}$
      \EndWhile
      \State \Return{$\mathit{TS}$}
    \EndFunction
    \Statex
     \Function{Pred}{$s,\mathit{States},S,\bm{\pi},\mathbf{R}$}
       \State \textbf{if} $\bm{\pi}(s)>0$ \textbf{then} \Return{NIL} \textbf{end if}
        \For{$s'\in \mathit{States}$}
           \State \textbf{if} $\mathbf{R}(s',s)\!>\!0\wedge \forall s''\!\in\! S\!\setminus\! \{s,s'\}\, .$ 
            \Statex \hspace*{2.5cm} $(\mathrm{R}(s',s'')\!=\!0\; \wedge\; \mathrm{R}(s'',s)\!=\!0)$ \textbf{then}
            \State \hspace*{5mm} \Return{$s'$}
           \State \textbf{end if}
        \EndFor
        \State \Return{NIL}
     \EndFunction
    \Statex
     \Function{Succ}{$s,\mathit{States},S,\bm{\pi},\mathbf{R}$}
        \For{$s'\in \mathit{States}$}
           \State \textbf{if} $\bm{\pi}(s')=0\wedge\mathbf{R}(s,s')\!>\!0\wedge \forall s''\!\in\! S\!\setminus\! \{s,s'\}\, .$ 
            \Statex \hspace*{2.5cm} $(\mathrm{R}(s,s'')\!=\!0\; \wedge\; \mathrm{R}(s'',s')\!=\!0)$ \textbf{then}
            \State \hspace*{5mm} \Return{$s'$}
           \State \textbf{end if}
        \EndFor
        \State \Return{NIL}
     \EndFunction
  \end{algorithmic}
\end{algorithm}

Finally, the result in this section supports the calculation and exploitation of the ``together'' state sets 
from~(\ref{eq:Spartition}). 

\begin{definition}
\label{def:together-sets}
\revised{The \emph{together state sets} $S_1, S_2, \ldots, S_m$ for an until path formula $P_{=?}[\Phi_1 U^{I} \Phi_2]$ over the Markov chain $\mathcal{M}=\mathsf{CTMC}(S,\bm{\pi},\mathbf{R})$ are the state sets comprising the same elements as the $m$ state sequences returned by function \textsc{TogetherSeqs}($\mathsf{CTMC}(S,\bm{\pi},\mathbf{R})$, $S_\textsf{X}$, $S_\textsf{O}$) from Algorithm~\ref{alg:together}, where $S_\textsf{X}$ and $S_\textsf{O}$ are the exclude-from-refinement and once-only state sets for the formula.}
\end{definition}

The function \textsc{TogetherSeqs} 
builds the $m$ state sequences in successive iterations of its outer while loop (lines~3--26). The set $\mathit{States}$ maintains the states yet to be allocated to sequences (initially $S\setminus (S_\textsf{X}\cup\ S_\textsf{O})$, cf.~line~2), and each new sequence $T$ starts with a single element picked randomly from $\mathit{States}$ (line~4).
The inner while loop in lines~7--24 ``grows'' this sequence. First, the if statement in lines~8--15 tries to grow the sequence to the left with a state $s$ that ``precedes'' the sequence, in the sense that the only outgoing CTMC transition from $s$ is to the sequence head, and the only way of reaching the sequence head is through an incoming CTMC transition  from $s$. Analogously, the if statement in lines~16-23 grows the sequence to the right, by appending to it the state that ``succeeds'' the state at the tail of the sequence, if such a ``successor'' state exists. The predecessor and successor states of a state $s$ are computed by the functions \textsc{Pred} and \textsc{Succ}, respectively, where these functions return NIL if the states they attempt to find do not exist. The inner while loop terminates when the $\mathit{States}$ set becomes empty or  the sequence $T$ has no more predecessors or successors, so the flags $\mathit{left}$ and $\mathit{right}$ are set to $\mathit{false}$ in lines 13 and 21, respectively. On exit from this while loop, the sequence $T$ is added to the set of sequences $\mathit{TS}$, which is returned (line~27) after the outer while loop also terminates when $\mathit{States}$ becomes empty. Termination is guaranteed since at least one element is removed from $\mathit{States}$ in each iteration of this while loop (in line~5). 

To analyse the complexity of \textsc{TogetherSeqs}, we note that the worst case scenario corresponds to $S_\textsf{X}=S_\textsf{O}=\emptyset$ and to the function returning only sequences of length $1$, in which case the outer while loop is executed $|S|$ times with both \textsc{Pred} and \textsc{Succ} invoked once in each iteration. The if statements from \textsc{Pred} and \textsc{Succ} perform $\mathsf{O}(|S|)$ comparisons, and are executed within for loops with $\mathsf{O}(|S|)$ iterations, yielding an $\mathsf{O}(|S|^2)$ complexity for each function, and an overall $\mathsf{O}(|S|^3)$ complexity for the algorithm.

\begin{theorem}
\label{th:Sm}
If $T=\langle s_{i1}, s_{i2}, \ldots, s_{iN_i} \rangle$ is one of the sequences returned by \textsc{TogetherSeqs}, $\omega$ a path that satisfies $\Phi_1 U^{I} \Phi_2$ for an interval $I\!\subseteq\!\mathbb{R}_{\geq 0}$, and $t\!\in\! I$ the earliest time when $\omega@t\!\models\!\Phi_2$ (with $\omega@t'\!\models\!\Phi_1$ for all $t'\!\in\! [0,t)$), then up to time $t$ the states from $T$ can only appear on $\omega$  as complete sequences $\ldots s_{i1}t_{i1}s_{i2}t_{i2}\ldots s_{iN_i}t_{iN_i}\ldots$. 
\end{theorem}
\begin{proof}
The case $N_i=1$ is trivial, so we assume $N_i\!>\!1$ in the rest of the proof. We have two cases: either $\omega$ contains no states from $T$, or it contains at least one state from $T$. In the former case, the theorem is proven. In the latter case, consider any state $s_{i_j}$ that occurs on $\omega$, $1\!\leq\! j\!\leq\! N_i$. The states $s_{i1}$, $s_{i2}$, \ldots, $s_{i,{j-1}}$ must also occur on $\omega$, in this order and just before $s_{ij}$, as transitioning through each of these states is the only way to reach $s_{ij}$ in the CTMC. Moreover, $s_{i,{j+1}}$, $s_{i,{j+2}}$, \ldots, $s_{iN_i}$ must immediately follow $s_{ij}$ on $\omega$ (in this order) because $s_{ij}$ is not an absorbing state and its only outgoing transition is to $s_{i,{j+1}}$, etc. Hence, the path is of the form $\omega\!=\!s_1t_1s_2t_2\ldots s_xt_xs_{i1}t_{i1}s_{i2}t_{i2}\ldots s_{iN_i}t_{iN_i}\ldots$ for some $x\!\geq\! 0$. To prove that this occurrence of all states from $T$ on $\omega$ is either up to or after time $t$, we show that it is not possible to have $\omega@t\!=\!s_{ij}$ for any $j<N$. Indeed, if we assume $\omega@t\!=\!s_{ij}$ then according to the hypothesis $s_{ij}\models\Phi_2$  must hold. As this must be true not only for $\omega$ but also for any other path $\omega'$ that satisfies $\Phi_1 U \Phi_2$ and contains the states from $T$, definition~(\ref{eq:Sx}) implies that $s_{i,{j+1}},s_{i,{j+2}},\ldots,s_{iN_i}\in S_\textsf{X}$ because $\omega'$ comprises states that satisfy $\Phi_1$ followed by state $s_{ij}$ that satisfies $\Phi_2$, followed by the $s_{i,{j+1}},s_{i,{j+2}},\ldots,s_{iN_i}$. However, having states from $T$ in $S_\textsf{X}$ is not possible since line~2 of \textsc{TogetherSeqs} removes $S_\textsf{x}$ from the set of states used to generate $T$.
\end{proof}

Theorem~\ref{th:Sm} allows \acronym\ to model the delays of all states in the same ``together'' set $S_i$, $1\leq i\leq m$, as a \emph{joint delay}
\begin{equation}
  \label{eq:joint-delay}
  \Delta_i = \sum_{j=1}^{N_i}\delta_{ij}, 
\end{equation}
since the relevant part of any path that influences the value of $P_{=?}[\Phi_1 U^{I} \Phi_2]=0$ contains either all these states or none of them. 

\begin{example}
\label{ex3}
Consider again the QoS properties~(\ref{eq:web-reqs}) of the web application from our motivating example. For property \textbf{P2}, \textsc{TogetherSeqs} is called with $S_\textsf{X}=\{s_2,s_4,s_7\}$ (cf.~Example~\ref{ex1}) and $S_\textsf{O}=\{s_1,s_6\}$ (cf.~Example~\ref{ex2}), so it starts with $\mathit{States}=\{s_1,s_2,\ldots,s_7\}\setminus(S_\textsf{X}\cup S_\textsf{O})=\{s_3,s_5\}$ in line~2. Irrespective of which of $s_3$ and $s_5$ is picked in line~5, the other state will be added to the same ``together'' sequence since the two states always follow one another with no intermediate states. The ``together'' sets for the other two properties from~(\ref{eq:web-reqs}) are given in Table~\ref{tab:classification}, which brings together the results from Examples~\ref{ex1}--\ref{ex3}.
\end{example}

\begin{table}
\centering
\caption{CTMC state partition for the web application properties}
\label{tab:classification}
\begin{tabular}{cccc}\toprule
 \textbf{Property} & $S_\textsf{X}$ & $S_\textsf{O}$  & $S_1$, $S_2$, \ldots, $S_m$ \\ \toprule
 \textbf{P1} & $\{s_7\}$ & $\{s_1,s_6\}$ & $\{s_2,s_4\}$, $\{s_3,s_5\}$ \\
 \textbf{P2} & $\{s_2,s_4,s_7\}$ & $\{s_1,s_6\}$ & $\{s_3,s_5\}$ \\
 \textbf{P3} & $\{s_7\}$ & $\{s_1,s_6\}$ &$\{s_2,s_4\}$, $\{s_3,s_5\}$ \\ \bottomrule
 \end{tabular}
\end{table}

\subsection{Selective refinement}

The second \acronym\ step models the delays and holding times of the relevant components of the analysed system according to the rules established in the previous section. These rules are summarised in Table~\ref{table:refinementRules}, and require methods for joint delay modelling (for components associated with ``together'' CTMC states) and for individual holding time modelling (for components not associated with $S_\textsf{X}$ states). The two methods are described next.

\begin{table*}
\centering
\caption{\acronym\ rules for modelling the delays and holding times of different types of system components}
\label{table:refinementRules}
\begin{small}
\begin{tabular}{cccc}\toprule
& $S_\textsf{X}$ \textbf{components} & $S_\textsf{O}$ \textbf{components} & $S_i$ \textbf{components}, $1\leq i\leq m$ \\ \toprule
 \textbf{delay} & no modelling needed &  no modelling needed$^\dagger$ & joint delay modelling \\
 \textbf{holding time} & no modelling needed & per component modelling & per component modelling \\ \bottomrule
 \multicolumn{4}{l}{$^\dagger$$S_\textsf{O}$ components introduce a deterministic delay  $\sum_{s_i\in S_\textsf{O}}\delta_i$}
\end{tabular}
\end{small}
\end{table*}

\begin{figure}
\centering
\includegraphics[width=\hsize]{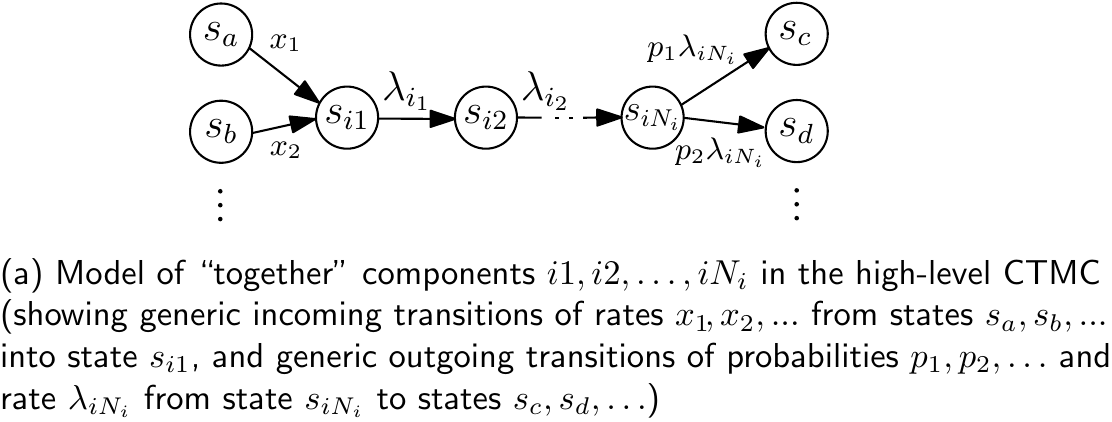} 

\vspace*{5mm}
\includegraphics[width=\hsize]{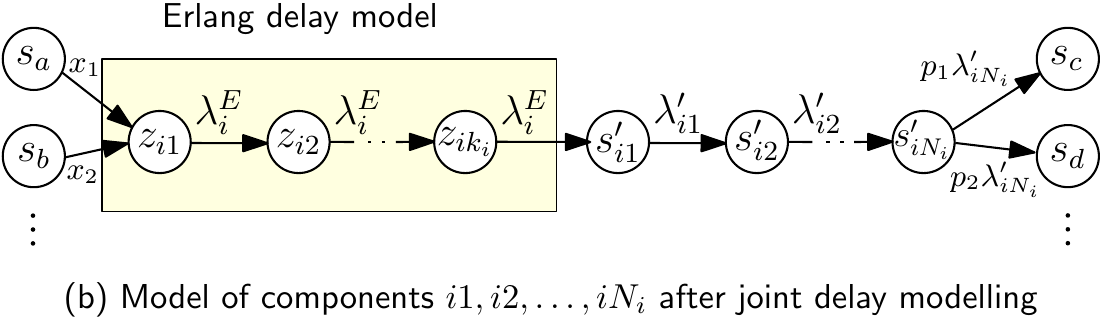}
\caption{Joint delay modelling for ``together'' state set $S_i$ whose final state $s_{iN_i}$ has $M_i\geq 1$ outgoing transitions}
\label{fig:delayExtraction}

\vspace*{-2mm}
\end{figure}

\subsubsection{Joint delay modelling}
\label{sec:delay}


%
%
%

For each ``together'' set $S_i\!=\!\{s_{i1}, s_{i2}, \ldots, s_{iN_i}\}$, \acronym\ extends the CTMC with additional states and transitions that model the joint delay $\Delta_i$ from~(\ref{eq:joint-delay}) by means of an Erlang distribution, i.e., a sum of several independent exponential distributions with the same rate~\cite{stewart2009probability}. As shown in Fig.~\ref{fig:delayExtraction}, this involves replacing the states from $S_i$ with a sequence of delay-modelling states $z_{i1}$, $z_{i2}$, \ldots, $z_{ik_i}$ that encode an Erlang-$k_i$ distribution of rate $\lambda_i^E$, followed by states $s'_{i1}, s'_{i2}, \ldots, s'_{iN_i}$ with  transitions matching those from the high-level CTMC but of rates $\lambda'_{i1}$, $\lambda'_{i2}$, \ldots, $\lambda'_{iN}$. 

However, delays are not modelled perfectly by Erlang distributions: for any \emph{error} $\epsilon\!\in\!(0,1)$, there is a (small) probability $p$ that the refined CTMC leaves state $z_{ik_i}$ within $\Delta_i(1\!-\!\epsilon)$ time units of entering $z_{i1}$. Given specific values for $\epsilon$ and $p$, the theorem below supports the calculation of the parameters $k_i$, $\lambda_i^E$ and $\lambda'_{i1}$ to $\lambda'_{iN_i}$ for our joint delay modelling.


\begin{theorem}
\label{th:erlang}
Given an error bound $\epsilon\in (0,1)$, if the joint delay modelling parameters $k_i$, $\lambda_i^E$ and $\lambda'_{i1}$ to $\lambda'_{iN_i}$ satisfy
\begin{equation}
\label{eq:th1}
\begin{array}{c}
   \textrm{(a) } 1- \sum_{l=0}^{k_i-1}\frac{\left(k_i(1-\epsilon)\right)^l e^{-k_i(1-\epsilon)}}{l!}  = p \\[2mm]
   \textrm{(b) } \lambda_i^E = \frac{k_i}{\Delta_i}\\[2mm]
   \textrm{(c) } \forall j\in\{1,2,\ldots,N_i\}\,.\,\lambda'_{ij} = \frac{\lambda_{ij}}{1-\lambda_{ij}\delta_{ij}}
\end{array}
\end{equation}
for some value $p\in (0,1)$ then the following properties hold for the refined CTMC:
\begin{itemize}
\item[(i)] The probability that the CTMC leaves state $z_{ik_i}$ within $\Delta_i(1-\epsilon)$ time units from entering state $z_{i1}$ is $p$;
\item[(ii)] The expected time for the refined CTMC to leave $s'_{iN_i}$ after entering state $z_{i1}$ is $\sum_{j=1}^{N_i} \lambda_{ij}^{-1}$. This is also the expected time for the high-level CTMC to leave state $s_{iN_i}$ after entering state $s_{i1}$, so the joint delay modelling preserves the first moment of the distribution associated with the refined CTMC states.
\end{itemize}
\end{theorem}
\begin{proof}
To prove (i), \revised{recall from Section~\ref{sect:Erlang}} that the cumulative distribution function of an Erlang-$k$ distribution with rate $\lambda$ is 
$ F(k,\lambda,x) = 1- \sum_{l=0}^{k-1}\frac{(\lambda x)^l e^{-\lambda x}}{l!}$, so (\ref{eq:th1}a) can be rewritten as $F(k_i,\lambda_i^E,\Delta_i(1-\epsilon))=p$ since $k_i=\lambda_i^E\Delta_i$ according to~(\ref{eq:th1}b). Therefore, the probability that the Erlang delay model from Fig.~\ref{fig:delayExtraction} will transition from entering state $z_{i1}$ to exiting state $z_{ik_i}$ within $\Delta_i(1-\epsilon)$ time units is $p$. 
For part (ii), the expected time for the CTMC to leave state $s'_{iN_i}$ after entering $z_{i1}$ is the sum of the mean of the Erlang-$k_i$ distribution with rate $\lambda_i^E$ and the mean of the exponential distributions with rates $\lambda'_{i1}$ to $\lambda'_{iN_i}$, i.e.
\[
\begin{array}{l}
    k_i \frac{1}{\lambda_i^E} + \sum_{j=1}^{N_i}\frac{1}{\lambda'_{ij}} = \Delta_i + \sum_{j=1}^{N_i} \frac{1-\lambda_{ij}\delta_{ij}}{\lambda_{ij}}=\\
    \qquad =\sum_{j=1}^{N_i} \delta_{ij}+ \sum_{j=1}^{N_i} \left(\frac{1}{\lambda_{ij}}-\delta_{ij}\right) = \sum_{j=1}^{N_i} \frac{1}{\lambda_{ij}}.
\end{array}    
\]
\end{proof}

Theorem~\ref{th:erlang} supports the calculation of the delay model parameters for a ``together'' state set $S_i\!=\!\{s_{i1}, s_{i2}, \ldots, s_{iN_i}\}$ as follows:\\[-6mm]
\begin{enumerate}
\item Approximate the delays $\delta_{i1}, \delta_{i2}, \ldots, \delta_{iN_i}$ for the components associated with each state from $S_i$ using~(\ref{eq:delay}).
\item Compute the joint delay $\Delta_i=\sum_{j=1}^{N_i} \delta_{ij}$.
\item Choose a small error $\epsilon\!\in\!(0,1)$ and a small probability $p$ (e.g.\ $\epsilon\! =\! 0.1$ and $p\!=\!0.05$), and solve~(\ref{eq:th1}a) for $k_i$. This can be done using a numeric solver and rounding the result up to an integer value or, since $k_i$ only depends on $\epsilon$ and $p$, and is independent of $\Delta_i$, by using precomputed $k_i$ values as in Table~\ref{table:precomputed};
\item Calculate $\lambda_i^E$ and $\lambda'_{i1}$ to $\lambda'_{iN_i}$ using (\ref{eq:th1}b) and (\ref{eq:th1}c), respectively.
\end{enumerate}

\begin{table}
\centering
\caption{Precomputed $k_i$ values used in the experiments from Section~\ref{section:evaluation} \label{table:precomputed}}
\includegraphics[width=0.58\hsize]{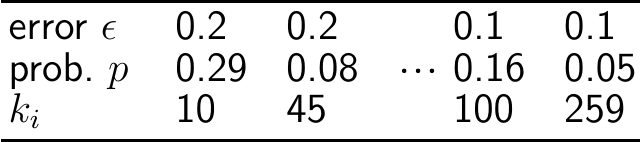}
\end{table}

%

\begin{example}
Consider the ``together'' set $S_1=\{s_2,s_4\}$ for property \textbf{P1} from our running example (cf.~Table~\ref{tab:classification}). States $s_2$ and $s_4$ correspond to the invocations of the arrivals and search web services from the travel web application, which according to our experimental data  have delays $\delta_2=45$ms and $\delta_4=209$ms, respectively. Therefore, the joint delay is $\Delta_1=\delta_2+\delta_4=254$ms. Suppose that we want to model this joint delay with an error bound $\epsilon\!=\!0.1$ and a probability $p\!=\!0.05$. This gives $k_1\!=\!259$ (cf.\ Table~\ref{table:precomputed}) and the other joint delay modelling parameters are calculated as: $\lambda_1^E = \frac{k_1}{\Delta_1} =\frac{259}{0.254}=1019$s$^{-1}$, $\lambda'_2=\frac{\lambda_2}{1-\lambda_2\delta_2}=$ $\frac{19.88}{1-19.88\cdot0.045}=188.61$s$^{-1}$ and $\lambda'_4=\frac{\lambda_4}{1-\lambda_4\delta_4}=$ $\frac{1.85}{1-1.85\cdot0.209}=3.01$s$^{-1}$ (where the rates $\lambda_2$ and $\lambda_4$ are taken from Table~\ref{tab:webservices}).
\end{example}

The next theorem gives the format of the refined CTMC after joint delay modelling is applied to all ``together'' state sets $S_1$ to $S_m$. 

\begin{theorem}
\label{th:after-delay}
Applying the \acronym\ joint delay modelling procedure to the ``together'' state set $S_i$ of a high-level model $\mathsf{CTMC}(S,\bm{\pi},\mathbf{R})$ yields a model $\mathsf{CTMC}(S',\bm{\pi}',\mathbf{R}')$ with:
\[
\begin{array}{l}
 S' \!\!=\! (S\setminus S_i)\cup  \{z_{i1},z_{i2}, \ldots, z_{ik_i}, s'_{i1},s'_{i2},\ldots,s'_{iN_i} \};\\ \\
 \bm{\pi}'(s) \!=\! 
\left\{\!\!
\begin{array}{ll}
\bm{\pi}(s), & \textrm{if } s\in S\setminus S_i\\
\bm{\pi}(s_{i1}), & \textrm{if } s=z_{i1}  \qquad ; \\
0, & \textrm{otherwise}
\end{array}
\right.\\ \\
\end{array}
\]

\vspace*{-6mm}
\[
\begin{array}{l}
\mathbf{R}'(s,u)\!=\! 
\left\{\!\!
\begin{array}{ll}
\mathbf{R}(s,u), &\!\!\! \textrm{if } s,u\!\in\! S\setminus S_i\\
\mathbf{R}(s,s_{i1}), &\!\!\! \textrm{if } s\!\in\! S\setminus S_i \wedge u\!=\!z_{i1}\\
\lambda_i^E, &\!\!\!  \textrm{if } (s,u)\!\in\!\{(z_{i1},z_{i2}),\ldots,\\ & \;\;(z_{i,k_i-1},z_{ik_i}), (z_{ik_i},s'_{i1})\}\\ 
\lambda'_{ij}, &\!\!\! \textrm{if } s=s'_{ij} \wedge u=s'_{i,j+1},\\ & \;\;\qquad\qquad 1\leq j\leq N_i-1\\
\frac{\mathbf{R}(s_{iN_i},u)}{\lambda_{iN_i}}\lambda'_{iN_i}, &\!\!\! \textrm{if } s\!=\!s'_{iN_i}\wedge u\!\in\! S\setminus S_i\\
\frac{\mathbf{R}(s_{iN_i},s_{i1})}{\lambda_{iN_i}}\lambda'_{iN_i}, &\!\!\! \textrm{if } s\!=\!s'_{iN_i}\wedge u\!=\!z_{i1}\\
0, &\!\!\! \textrm{otherwise}
\end{array}
\right.
\end{array}
\]
if $s\neq u$; and $\mathbf{R}'(s,s)=-\sum_{u\in S'\setminus\{s\}}\mathbf{R}'(s,u)$, where the terms $\frac{\mathbf{R}(s_{iN_i},u)}{\lambda_{iN_i}}$ and $\frac{\mathbf{R}(s_{iN_i},s_{i1})}{\lambda_{iN_i}}$ correspond to the probabilities $p_1,p_2,\ldots$ from Fig.~\ref{fig:delayExtraction} and are obtained using~(\ref{eq:pij}).
\end{theorem}
\begin{proof}
The proof is by construction, cf.\ Fig.~\ref{fig:delayExtraction}.
\end{proof}


\subsubsection{Holding-time modelling}

As indicated in Table~\ref{table:refinementRules}, we model the holding times of system components associated with high-level CTMC states from $S_\textsf{O}\cup S_1\cup S_2\ldots\cup S_m$ individually. For each such component, we synthesise a phase-type distribution $\mathsf{PHD}(\bm{\pi}_0,\mathbf{D}_0)$ that models the holding times~(\ref{eq:holding-times}), and we replace the relevant state $s'$ of the model $\mathsf{CTMC}(S',\bm{\pi}',\mathbf{R}')$ 
obtained after the \acronym\ joint delay modelling 
with this PHD. For operations corresponding to states $s_\textsf{O}\in S_\textsf{O}$ the replaced state is $s'=s_\textsf{O}$, while for operations corresponding to a state $s_{ij}\in S_i$, $1\leq i\leq m, 1\leq j\leq N_i$, the replaced state is the state $s'=s'_{ij}$ obtained after the joint delay modelling of $S_i$ (cf.~Fig.~\ref{fig:delayExtraction}b). 

Our holding-time modelling exploits recent advances in the fitting of phase-type distributions to empirical data. 
Given the usefulness of PHDs in performance engineering, this area has received considerable attention~\cite{buchholz2014phase,Okamura2016}, with effective PHD fitting algorithms developed based on techniques such as moment matching \cite{STM-200056210,15326340701300712}, expectation maximisation \cite{asmussen1996fitting,1673383,Wang2008} and Bayes estimation \cite{doi:10.1080/03610918.2013.848895,doi:10.1080/03610926.2010.483306}. Recently, these algorithms have been used within PHD fitting approaches that: (a)~partition the dataset into segments \cite{Wang2008} or clusters \cite{reinecke2012cluster} of ``similar'' data points; (b)~employ an established algorithm to fit a PHD with a simple structure to each data segment or cluster;  and (c)~use these simple PHDs as the branches of a PHD that fits the whole dataset. These approaches achieve better trade-offs between the size, accuracy and complexity of the final PHD than the direct algorithms applied to the entire dataset. 

\begin{algorithm}[t]
  \caption{Holding-time modelling with parameters:\\ \hspace*{3mm}$\bullet$ $\mathit{MinC}$ --- minimum number of PHD clusters\\ \hspace*{3mm}$\bullet$ $\mathit{MaxC}$ --- maximum number of PHD clusters\\ \hspace*{3mm}$\bullet$ $\mathit{MaxP}$ --- maximum number of cluster phases\\ \hspace*{3mm}$\bullet$ $\mathit{FittingAlg}$ --- basic PHD fitting algorithm\\ \hspace*{3mm}$\bullet$ $\mathit{MaxSteps}$ --- maximum steps without improvement \label{alg:htmodelling}}

  \begin{algorithmic}[1]
    \Function{$\!\,$HoldingTimeModeling}{$\alpha,\tau'_{i1},\tau'_{i2},\ldots,\tau'_{in_i}\!$}
          \State $\mathit{sample}\gets (\tau'_{i1},\tau'_{i2},\ldots,\tau'_{in_i})$
          \State $\mathit{minErr}=\infty$
          \State $\mathit{improvement} \gets 0$
          \State $\mathit{steps} \gets 0$
          \State $c \gets \mathit{MinC}$
          \While {$c\leq \mathit{MaxC} \wedge \mathit{steps} \leq \mathit{MaxSteps}$}
             \State $\mathit{phd} \gets \textsc{CBFitting}(\mathit{sample}, c, \mathit{FittingAlg}, \mathit{MaxP})$
             \State $err \gets \Delta \mathsf{CDF}(\mathit{sample}, \mathit{phd} )$
             \If {$\mathit{err}<\mathit{minErr}$}
		 \State $\mathit{best\_phd} \gets \mathit{phd}$
                  \State $\mathit{improvement} \gets \mathit{improvement}\! +\! (\mathit{minErr}\!-\!\mathit{err})$ 
    	  	 \State $\mathit{minErr} \gets \mathit{err}$
    	     \EndIf
             \If {$\mathit{improvement}\geq\alpha$}	   
		 \State $\mathit{improvement} \gets 0$
		 \State $\mathit{steps} \gets 0$
             \Else
                 \State $\mathit{steps} \gets \mathit{steps}+1$
    	     \EndIf
             \State $c \gets c+1$
          \EndWhile
          \State \Return $\mathit{best\_phd}$
          \EndFunction
  \end{algorithmic}
\end{algorithm}

The \acronym\ \textsc{HoldingTimeModeling} function from Algorithm~\ref{alg:htmodelling} achieves similar benefits by employing Reinecke et al.'s \emph{cluster-based PHD fitting} approach \cite{reinecke2012cluster,reinecke2012hyperstar,reinecke2013phase} to fit a PHD to the holding time sample $\tau'_{i1},\tau'_{i2},\ldots,\tau'_{in_i}$ from~(\ref{eq:holding-times}). The PHD fitting is carried out by the while loop in lines 7--22, which iteratively assesses the suitability of PHDs obtained when partitioning the $\mathit{sample}$ assembled in line~2 into $c=\mathit{MinC},\mathit{MinC}+1,\ldots,\mathit{MaxC}$ clusters. Line~8 obtains a PHD with $c$ branches (corresponding to partitioning $\mathit{sample}$ into $c$ clusters) and up to $\mathit{MaxP}$ phases by using the function \mbox{\textsc{CBFitting}}, which implements the cluster-based PHD fitting from \cite{reinecke2012cluster}. The $\mathit{FittingAlg}$ argument of \mbox{\textsc{CBFitting}} specifies the basic PHD fitting algorithm applied to each cluster as explained above, and can be any of the standard moment matching, expectation maximisation or Bayes estimation PHD fitting algorithms. The quality of the $c$-branch PHD is assessed in line~9 by using the CDF-difference metric~\cite{reinecke2012cluster} to compute the difference $\mathit{err}$ between $\mathit{sample}$ and the PHD. The if statement in lines~10--14 identifies the PHD with the lowest $\mathit{err}$ value so far, retaining it in line~11. Reductions in $\mathit{err}$ (i.e., ``improvements'') are cumulated in $\mathit{improvement}$ (line~12), and the while loop terminates early if the iteration counter $\mathit{steps}$ exceeds $\mathit{MaxSteps}$ before $\mathit{improvement}$ reaches the threshold $\alpha\geq 0$ provided as an parameter to \textsc{HoldingTimeModeling} and the $steps$ counter is reset in line~17. Finally, the best PHD achieved within the while loop is returned in line~23.

\begin{theorem}
Using a phase-type distribution $\mathsf{PHD}(\bm{\pi}_0,\mathbf{D}_0)$ generated by Algorithm~\ref{alg:htmodelling} to apply the \acronym\ holding-time modelling procedure to the state $s'$ of a model $\mathsf{CTMC}(S',\bm{\pi}',\mathbf{R}')$ yields a model $\mathsf{CTMC}(S'',\bm{\pi}'',\mathbf{R}'')$ with:
\[
\begin{array}{l}
 S''\!=\! (S'\setminus \{s'\})\cup  \{w_1,w_2,\ldots,w_N\}\; ;\\ \\
 \bm{\pi}''(s) \!=\! 
\left\{\!\!
\begin{array}{ll}
\bm{\pi}'(s), & \textrm{if } s\in S\!\setminus\! \{s'\}\\
\bm{\pi}'(s')\bm{\pi}_0(w_i), & \textrm{if } s=w_i, 1\!\leq\! i\!\leq\! N  \; ; \\
0, & \textrm{otherwise}
\end{array}
\right.\\ \\
\mathbf{R}''(s,u)\!=\! 
\left\{\!\!
\begin{array}{ll}
\mathbf{R}'(s,u), &\!\!\! \textrm{if } s,u\!\in\! S'\!\setminus\! \{s'\}\\
\mathbf{D}_0(s,u), &\!\!\! \textrm{if } s,u\!\in\! \{w_1,w_2,\ldots,w_N\}\\
\mathbf{R}'(s,s')\bm{\pi}_0(w_i), &\!\!\! \textrm{if } s\!\in\! S'\!\setminus\! \{s'\} \wedge u\!=\!w_i,\\ & \;\;\qquad\qquad\qquad 1\leq i\leq N\\
\frac{\mathbf{R}'(s',u)}{\lambda'}\mathbf{d}_1(w_i), &\!\!\! \textrm{if } s\!=\!w_i \wedge u\!\in\! S'\!\setminus\! \{s'\},\\ & \;\;\qquad\qquad\qquad 1\leq i\leq N\\
\end{array}
\right.
\end{array}
\]
if $s\neq u$; and $\mathbf{R}''(s,s)=-\sum_{u\in S''\setminus\{s\}}\mathbf{R}''(s,u)$, where:
\squishlist
\item $w_1,w_2,\ldots,w_N$ are the transient states of $\mathsf{PHD}(\bm{\pi}_0,\mathbf{D}_0)$;
\item $\lambda'$ is the total outgoing transition rate for $s'$;
\item $\mathbf{d_1}=-\mathbf{D}_0\mathbf{1}$.
\squishend 
\end{theorem}
\begin{proof}
The proof is by construction, cf.\ Algorithm~\ref{alg:htmodelling}.
\end{proof}

\begin{figure*}
\centering
\includegraphics[width=0.3\linewidth]{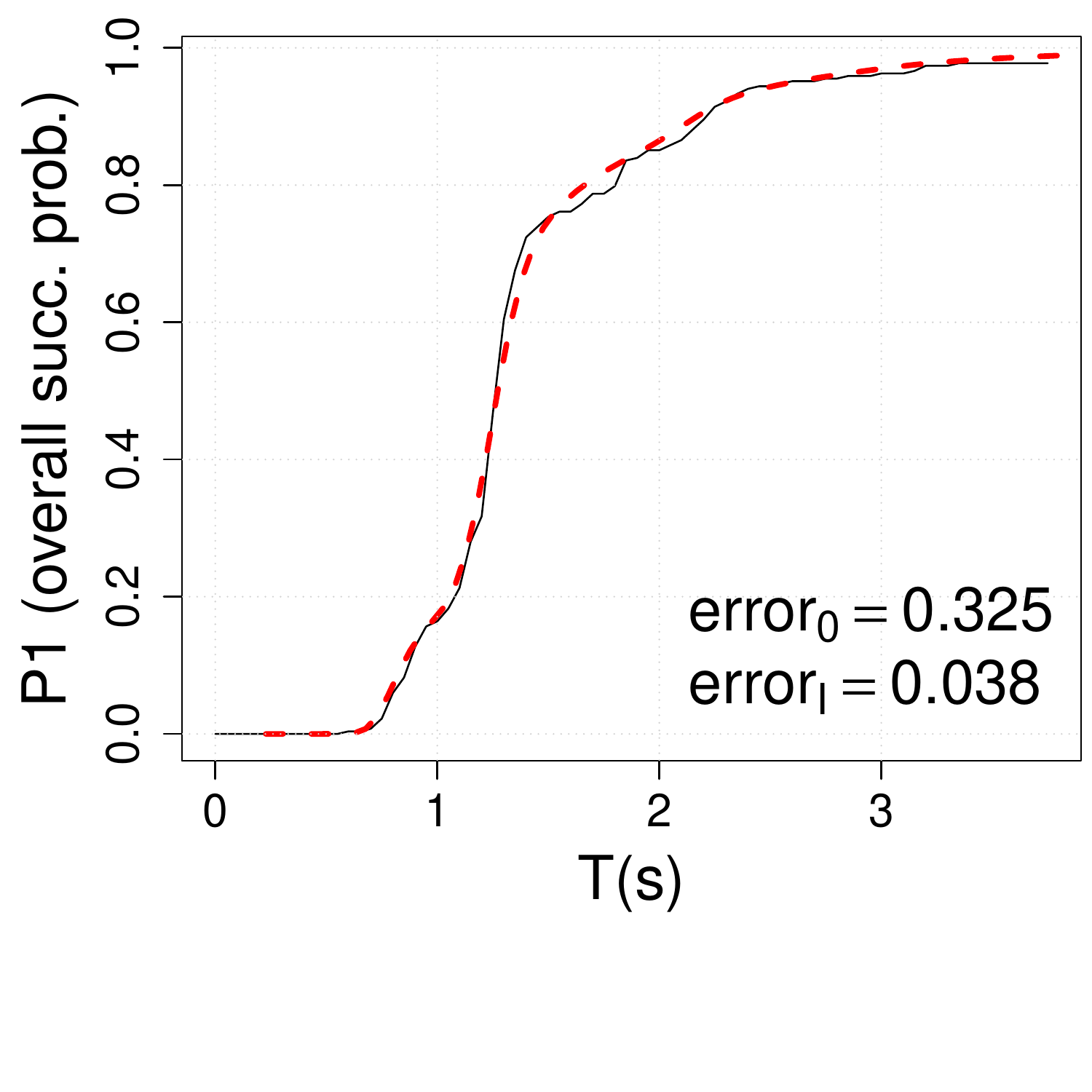}\hspace*{1mm}
\includegraphics[width=0.3\linewidth]{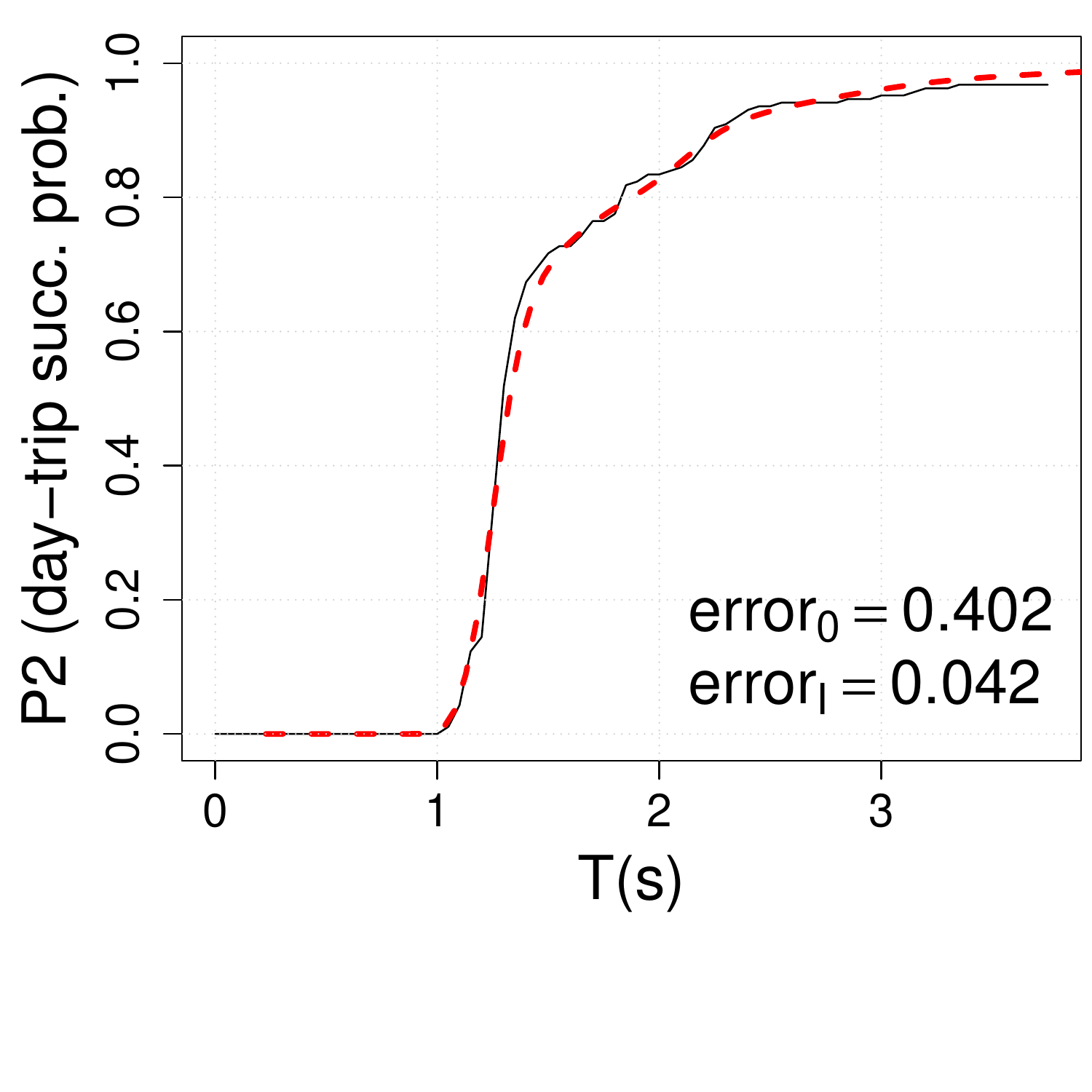}\hspace*{1mm}
\includegraphics[width=0.3\linewidth]{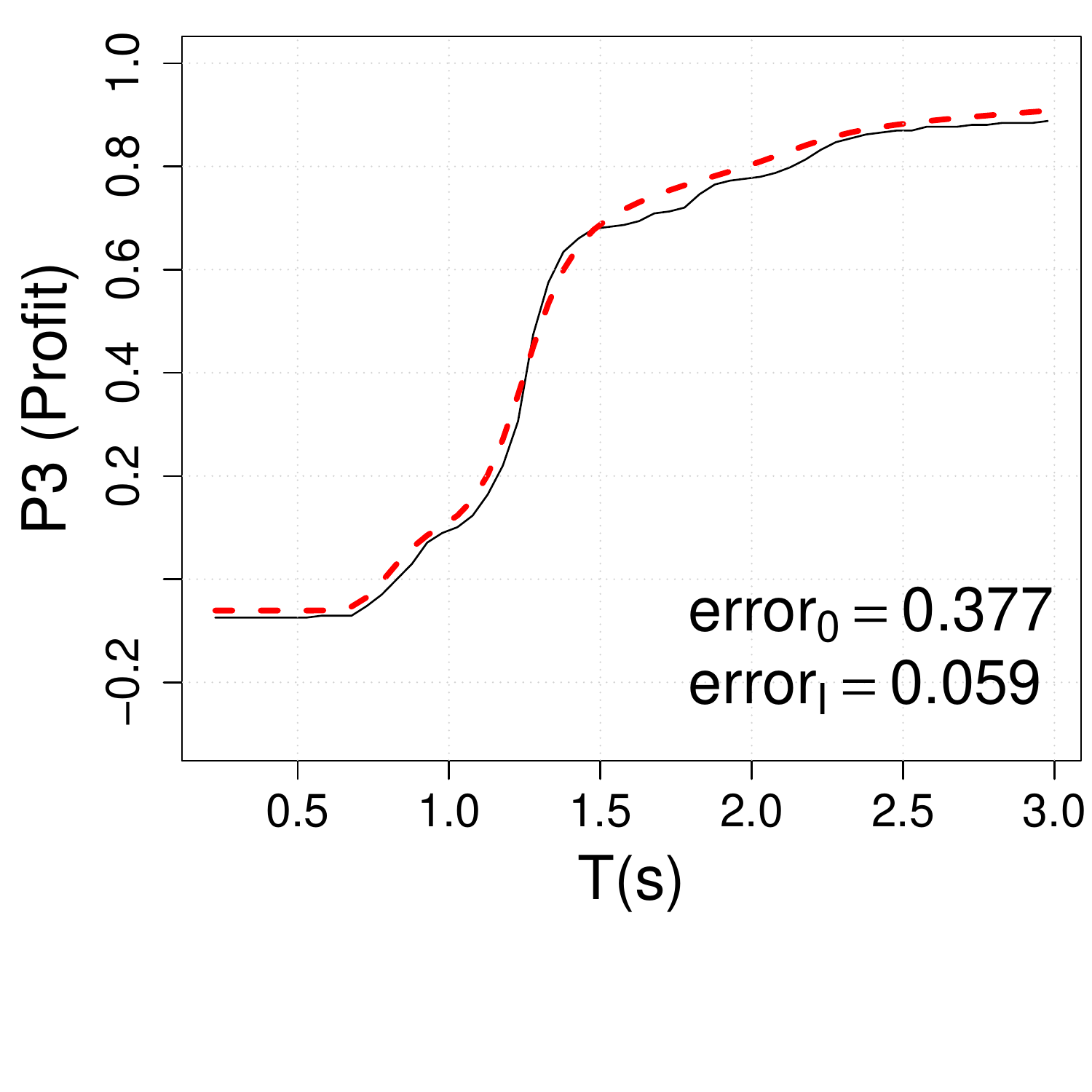}

\vspace*{-0.9cm}
\caption{Actual (continuous lines) property values versus property values predicted (dashed lines) using the \acronym-refined CTMC model; \textsf{error}$_0$ and \textsf{error}$_\textsf{I}$ represent the error~(\ref{eq:err}) for the high-level CTMC and the refined CTMC, respectively.}
\label{fig:VerRef}

\vspace*{-2mm}
\end{figure*}

\begin{example}
\label{ex:holding-state-refinement}
We used our \acronym\ refinement tool (Section~\ref{section:tool}) to perform the component classification and selective refinement steps of our approach on the high-level CTMC from the motivating example. Algorithm~\ref{alg:htmodelling} was executed for each component associated with a CTMC state from the $S_\textsf{O}$ or the $S_1$ to $S_m$ state sets in Table~\ref{tab:classification},  with $\alpha\!=\!0.1$ and with the configuration parameters $\mathit{MinC}\!=\!2$, $\mathit{MaxC}\!=\!30$, $\mathit{MaxP}\!=\!300$, $\mathit{MaxSteps}\!=\!3$ and $\mathit{FittingAlg}$ an expectation-maximisation PHD fitting algorithm that produces hyper-Erlang distributions.\footnote{A hyper-Erlang distribution \cite{Anichkin1983,buchholz2014phase,Wang2008} is a PHD in which the $c>1$ branches of the PHD from Algorithm~\ref{alg:htmodelling} are mutually independent Erlang distributions.} We obtained refined CTMCs comprising 730 states and 761 transitions for property \textbf{P1}, 367 states and 387 transitions for property \textbf{P2}, and 730 states and 761 transitions for property \textbf{P3}. Fig.~\ref{fig:VerRef} compares the actual values of properties \textbf{P1}--\textbf{P3} with the values predicted by the analysis of these refined CTMCs. 
Both the visual assessment and the error values \textsf{error}$_\mathsf{I}$ associated with these predictions (which are significantly lower than the error values \textsf{error}$_\mathsf{0}$ before refinement) show that \acronym\ supports the accurate analysis of the three properties. \revised{In fact, the predicted and actual values for all properties may seem surprisingly close. The explanation for this close match is twofold. First, the \acronym\ refinement uses PHD distributions, which -- if sufficiently large -- can approximate arbitrarily close any continuous distribution (cf.\ Section~\ref{sect:PHD}). Second, the apparently ''perfect'' match between the predicted and the actual QoS property values is slightly deceptive: for instance, a closer inspection of the results shows that there are still multiple points where the difference between the two is at least 5\%. In Section~\ref{subsect:rq1} we supplement this brief discussion of the results from Fig.~\ref{fig:VerRef} with an experimental evaluation which shows that the accurate \acronym\ results are not due to overfitting.} 
\end{example}


\section{\acronym\ Refinement Tool}
\label{section:tool}

We implemented \acronym\ as a Java tool that takes as input a high-level CTMC model. This model is specified in a  variant of the PRISM modelling language \cite{kwiatkowska2011prism} where state transition commands are expressed using components labels. For example, the PRISM command

\begin{small}
\[
   \mathsf{s\!=\!1} \,\rightarrow\, \mathsf{p_1\!*\!\lambda_1\!:\!(s'\!=\!2) + (1\!-\!p_1)\!*\!\lambda_1\!:\!(s'\!=\!3);}
\]
\end{small}

\vspace*{-3mm}
\noindent
that defines the outgoing transitions for state $s_1$ of our high-level CTMC model from Fig.~\ref{fig:CTMC} is replaced by

\begin{small}
\[
  \mathsf{s\!=\!\langle \textit{location}\rangle\rightarrow\, p_1\!:\!(s'\!=\!\langle \textit{arrival}\rangle) + (1\!-\!p_1)\!:\!(s'\!=\!\langle \textit{departures}\rangle);}
\]
\end{small}

\vspace*{-3mm}
\noindent
in the \acronym\ variant of the modelling language. This indicates that the CTMC transitions from the state associated with the \textit{location} component to either the state associated with the \textit{arrival} component (with probability $p_1$) or to the state associated with the \textit{departures} component (with probability $1-p_1$). An XML configuration file is then used to map each of these component labels to a file of comma-separated values containing the observed execution times for the relevant component. In addition, this configuration file allows the user to define the \acronym\ refinement parameters (i.e.\ $k_i$ from Theorem~\ref{th:erlang}, and $\alpha$, $\mathit{MinC}$, $\mathit{MaxC}$, $\mathit{MaxP}$ and $\mathit{MaxSteps}$ from Algorithm~\ref{alg:htmodelling}). The $\mathit{FittingAlg}$ parameter is fixed in the current version of the tool, so that \acronym\ uses the expectation-maximisation PHD fitting algorithm mentioned in Example~\ref{ex:holding-state-refinement}. 

\revised{When multiple QoS properties (for the same high-level CTMC) are provided to the \acronym\ tool, we avoid the overheads associated with the repeated execution of the modelling tasks from Table~\ref{table:refinementRules} for the same components by maintaining a cache of all completed tasks and their results. As such, each of these tasks is executed at most once per system component, and its cached result is used when needed instead of repeating the task. By comparison, refining the whole CTMC indiscriminately for even a single QoS property would require the execution of these modelling tasks for every system component.} 

Finally, to support the scenario where the component delays~(\ref{eq:delay}) are negligible compared to the holding times~(\ref{eq:holding-times}), the configuration file allows the specification of a \emph{delay threshold}, and components with delays~(\ref{eq:delay}) below this threshold are not included in the joint delay modelling step of the \acronym\ refinement. We found experimentally that this leads to significant reductions in the size of the refined CTMC with no impact on the accuracy of the QoS analysis.

Our \acronym\ tool uses the HyperStar PHD fitting tool from~\cite{reinecke2012hyperstar} for the \textsc{CBFitting} function from Algorithm~2, and produces the refined CTMCs as standard PRISM models. The \acronym\ tool is freely available from our project webpage \textsf{\url{https://www.cs.york.ac.uk/tasp/OMNI/}}, together with detailed instructions and all the models and datasets from this paper.


\section{Evaluation
\label{section:evaluation}}

We evaluated \acronym\ by performing a set of experiments aimed at answering the following research questions.
\vspace*{1.5mm}

\noindent
\textbf{RQ1 (Accuracy/No overfitting):} How effective are \acronym\ models at predicting QoS property values for other system runs than the one used to collect the execution-time observation datasets for the refinement? 
\vspace*{1.5mm}

\noindent
\textbf{RQ2 (Refinement granularity):} What is the effect of varying the \acronym\ refinement granularity on the refined model accuracy, size and verification time?
\vspace*{1.5mm}

\noindent
\textbf{RQ3 (Training dataset size):} What is the effect of the training dataset size on the refined model accuracy?
\vspace*{1.5mm}

\noindent
\textbf{RQ4 (Component classification):} What is the benefit of using a component classification step within \acronym? 
\vspace*{1.5mm}

To assess the generality of \acronym, we carried out our experiments within two case studies that used real systems and datasets from different application domains. The first case study is based on the travel web application presented in Section~\ref{section:motivation} and used as a motivating example earlier in the paper. In the second case study, we applied \acronym\ to an IT support system. This system is introduced in Section~\ref{subsect:second-case-study}, followed by descriptions of the experiments carried out to address the four research questions in Sections~\ref{subsect:rq1}--\ref{subsect:rq4}.

\subsection{IT Support System \label{subsect:second-case-study}}

The real-world IT support system we used to evaluate \acronym\ is deployed at the Federal Institute of Education, Science and Technology of Rio Grande de Norte (IFRN), Brazil. The system enables the IFRN IT support team to handle user tickets reporting problems with the institute's computing systems. As part of our collaboration with IFRN researchers~\cite{da2017self}, system logs covering the handling of 1410 user tickets were collected from this IT support system over a period of six months between September 2016 and February 2017.  

\begin{figure}
\centering

\vspace*{-2mm}
\includegraphics[width=0.89\linewidth]{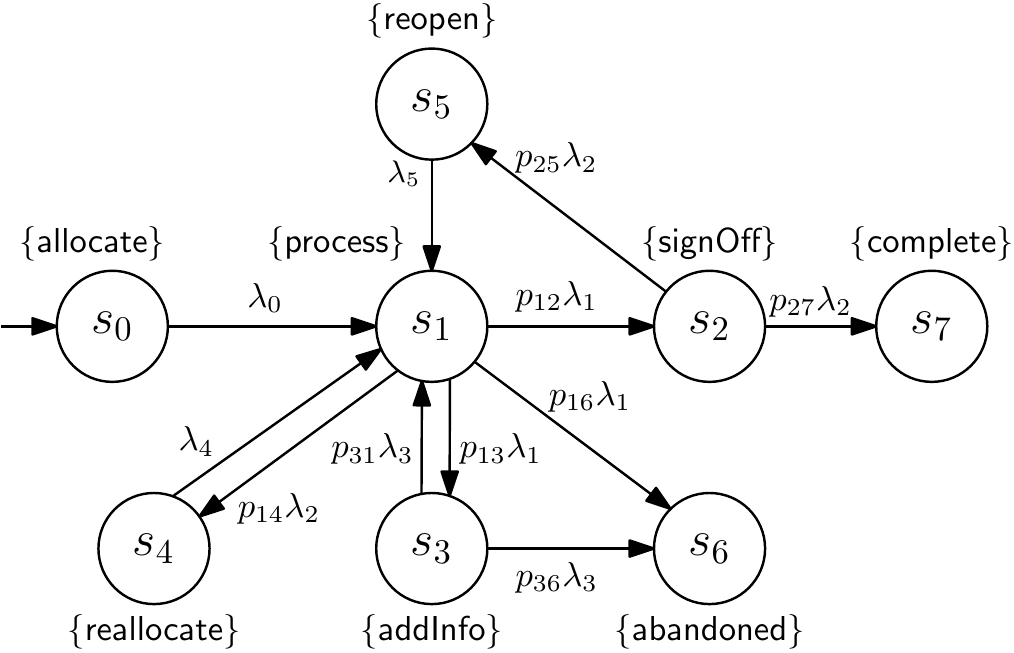}
\caption{High-level CTMC model of IT support system \label{fig:RB_Abstract}}
\end{figure}

A high-level CTMC model of the business process implemented by the IT system is shown in Fig.~\ref{fig:RB_Abstract}. In this model, state $s_0$ corresponds to a ticket being created by a ``client'' and awaiting allocation to a member of the support team. Once allocated, the ticket is processed (state $s_1$) and, if the issue can be resolved, the client is informed and the ticket awaits sign off ($s_2$) before being marked as complete ($s_7$). The client may choose to reopen ($s_5$) the ticket rather than close it, in which case the ticket is returned to the support team member for further processing. Whilst processing a ticket, the support staff may require additional information from the client ($s_3$) or may need to reallocate the ticket to another member of the IT support team ($s_4$). A ticket may also be abandoned ($s_6$) either during processing or whilst awaiting additional information from the client.

We used our \acronym\ tool to refine the high-level CTMC from Fig.~\ref{fig:RB_Abstract} \revised{ in order to support the verification of the response time of the IT support system through} the analysis of two properties:
\begin{equation}
\label{eq:reqs-IT-system}
\textrm{\hspace*{-0.3cm}}\begin{array}{ll}
\textbf{P1} & \,P_{=?}[F^{[0,T]} \mathit{complete}]\\
\textbf{P2} & \,P_{=?}[(\neg\mathit{reopen} ~\&~ \neg\mathit{addInfo})\; U^{[0,T]} \mathit{complete}] 
\end{array}
\end{equation}
where $\textbf{P1}$ specifies the probability of a ticket reaching the complete state within $T$ (working) hours, and $\textbf{P2}$ represents the probability of ticket handling being completed within $T$ working hours without further input from the client who raised the ticket and without the ticket being reopened.

We used only half of the six-month logs (covering 705 tickets created over the first approximately three months) for the \acronym\ refinement, so that we could use the other half of the logs to answer research question RQ1 (cf.~Section~\ref{subsubsect:rq1-IT}). For each ticket, the time spent in a particular state was derived from the log entries, taking into account only the working hours for the IT support team.\footnote{The working hours for the period covered by the logs were identified through consultation with the IFRN owner of the IT support process.} Assuming exponentially distributed execution times for the components of the IT support process, we used~(\ref{eq:rate}) to calculate the component execution rates shown in Table~\ref{tab:RBACRates}. Finally, we used the logs to calculate the frequencies of state transitions, and thus to estimate the CTMC state transition probabilities as shown in Table~\ref{tab:RBACParams}. 

\begin{table}
\centering
\caption{Execution rates for the IT support system \label{tab:RBACRates}}
\begin{tabular}{lc} \hline
\textbf{Component}  & \textbf{Rate} ($\mathsf{hours}^{-1}$) \\ \hline
allocate  & $\lambda_0=0.08248$  \\
process  & $\lambda_1=0.09799$  \\
signOff  & $\lambda_2=0.01167$ \\
addInfo & $\lambda_3=0.02006$ \\
reallocate & $\lambda_4=0.02839$ \\
reopen & $\lambda_5=0.09988$ \\ \hline
\end{tabular}
\end{table}

\begin{table}
\centering
\caption{Transition probabilities for the IT support system \label{tab:RBACParams}}
\begin{tabular}{ccccc} \hline
\multicolumn{2}{c}{\textbf{CTMC states}}  & $\!$\textbf{Transitions} & $\!\!\!\!$\textbf{Transitions} & $\!\!$\textbf{Estimate transition}  \\ \cline{1-2}
$s_i$ & $s_j$ & $\!\!$\textbf{from} $s_i$ \textbf{to} $s_j$ & \textbf{leaving} $s_i$ & \textbf{probability $s_i\rightarrow s_j$} \\
\hline
$s_1$ & $s_2$ & $533$ & $705$ & $p_{12}\!=\!\nicefrac{533}{705}\!=\!0.76$ \\
$s_1$ & $s_3$ & $24$   & $705$ & $p_{13}\!=\!\nicefrac{24}{705}\!=\!0.03$\\
$s_1$ & $s_4$ & $34$   & $705$ & $p_{14}\!=\!\nicefrac{34}{705}\!=\!0.05$ \\
$s_1$ & $s_6$ & $114$ & $705$ & $p_{16}\!=\!\nicefrac{114}{705}\!=\!0.16$ \\
$s_2$ & $s_5$ & $501$ & $533$ & $p_{25}\!=\!\nicefrac{501}{533}\!=\!0.94$ \\
$s_2$ & $s_7$ & $32$   & $533$ & $p_{27}\!=\!\nicefrac{32}{533}\!=\!0.06$ \\
$s_3$ & $s_1$ & $16$   & $24$   & $p_{31}\!=\!\nicefrac{16}{24}\!=\!0.67$ \\
$s_3$ & $s_6$ & $8$     & $24$   & $p_{36}\!=\!\nicefrac{8}{24}\!=\!0.33$ \\ \hline
\end{tabular}
\end{table}

Fig.~\ref{fig:VerRefTick} compares the actual values of properties \textbf{P1} and \textbf{P2} from~(\ref{eq:reqs-IT-system}) -- computed based on the system logs -- with the values predicted by the analyses of: (a)~the high-level CTMC from Fig.~\ref{fig:RB_Abstract} (with the parameters given in Tables~\ref{tab:RBACRates} and~\ref{tab:RBACParams}); and (b)~\acronym-refined CTMC models for the two properties. The refined CTMCs were obtained using the same \acronym\ parameters as in Example~\ref{ex:holding-state-refinement}, except $\alpha=0.2$ and a \emph{delay threshold} of 0.01 hours.\footnote{The threshold value was chosen to be approximately three orders of magnitude smaller than the smallest mean execution time of a system component, i.e.\ $1/\lambda_5=10.012$~hours for the IT support system (cf.~Table~\ref{tab:RBACRates}).} As explained in Section~\ref{section:tool}, this threshold meant that components with a delay~(\ref{eq:delay}) below 0.01~hours (which amounted to all component of the IT support system) were not included in the joint delay modelling of \acronym.

\begin{figure}
\centering
\mbox{\includegraphics[width=0.495\linewidth, height=4.8cm, trim=0.25cm 0 0 0,clip=true]{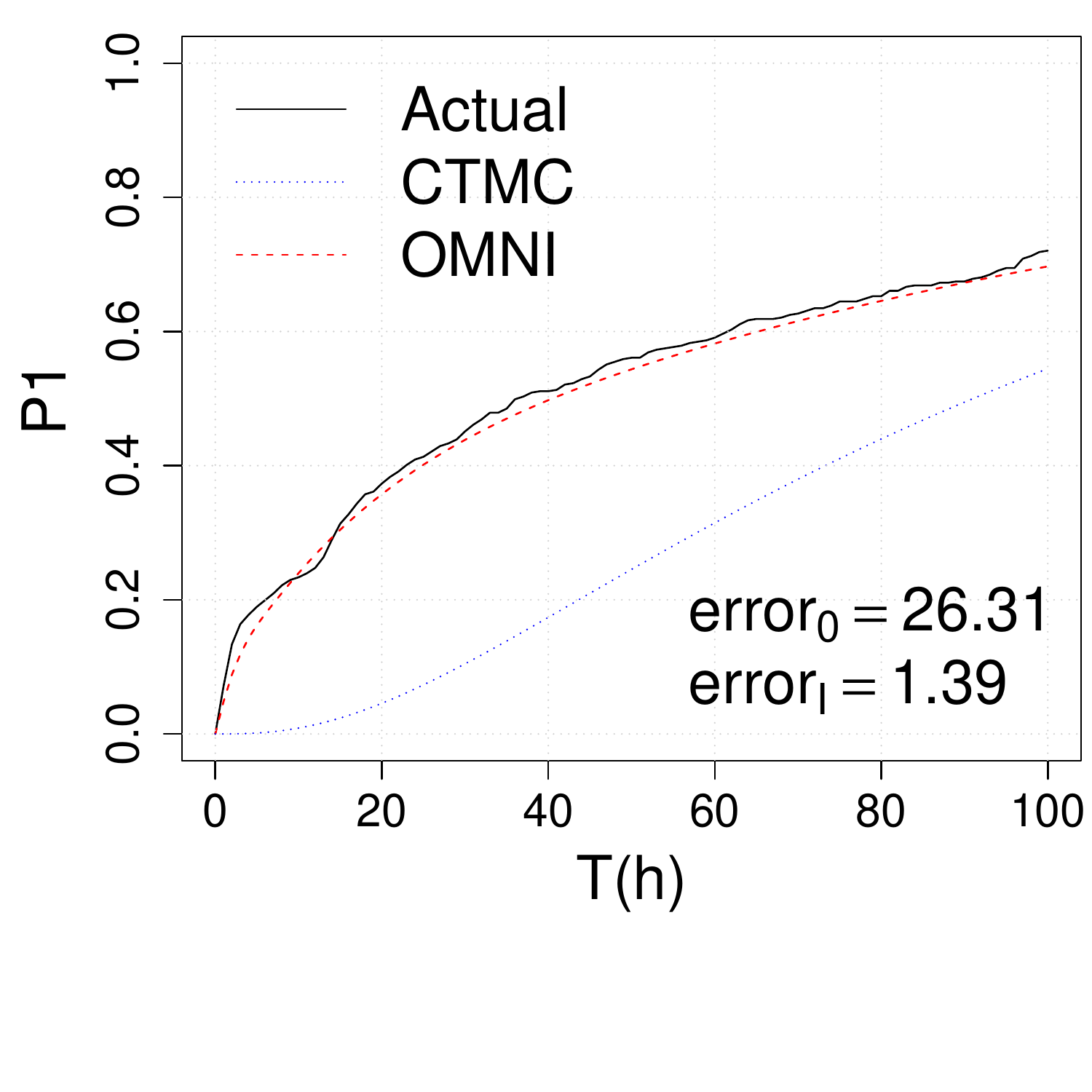} \includegraphics[width=0.495\linewidth, height=4.8cm, trim=0.25cm 0 0 0,clip=true]{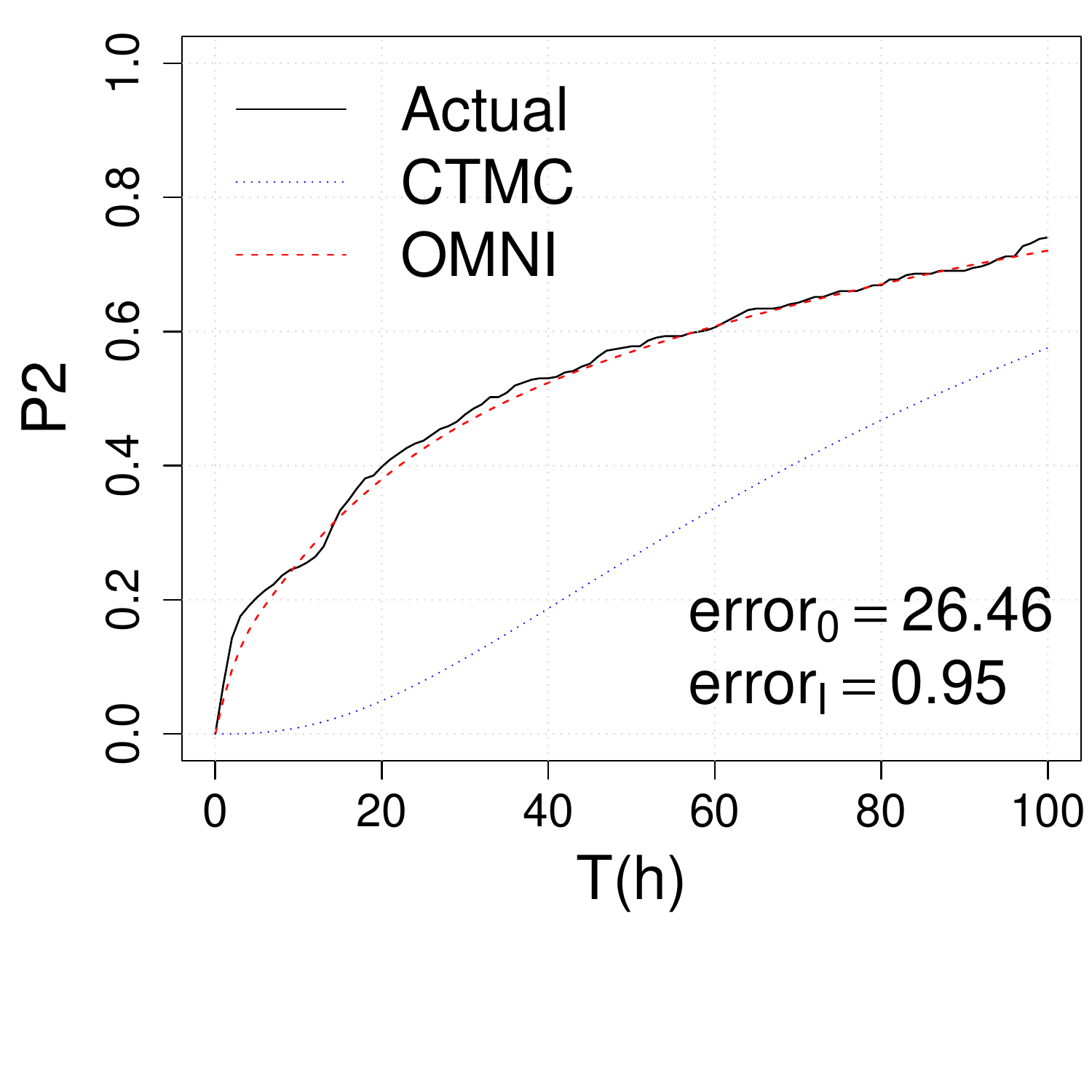}}

\vspace*{-8mm}
\caption{Actual values of the IT support system properties versus property values predicted using the high-level and the refined CTMC models, over 100 working hours from ticket creation; the prediction error~(\ref{eq:err}) for the refined CTMCs (i.e.\ $\mathsf{error_I}$) is 94.7\% smaller (for property \textbf{P1}) and 97\% smaller (for property \textbf{P2}) than the corresponding prediction errors for the high-level CTMC (i.e.\ $\mathsf{error_0}$).}
\label{fig:VerRefTick}
\end{figure}

Having introduced the system used in our second case study, we will use the next sections to describe the experiments carried out to answer our four research questions.

\begin{figure}
\centering
\includegraphics[width=\linewidth, height=6cm, trim=0.2cm 1.5cm 6.95cm 0,clip=true]{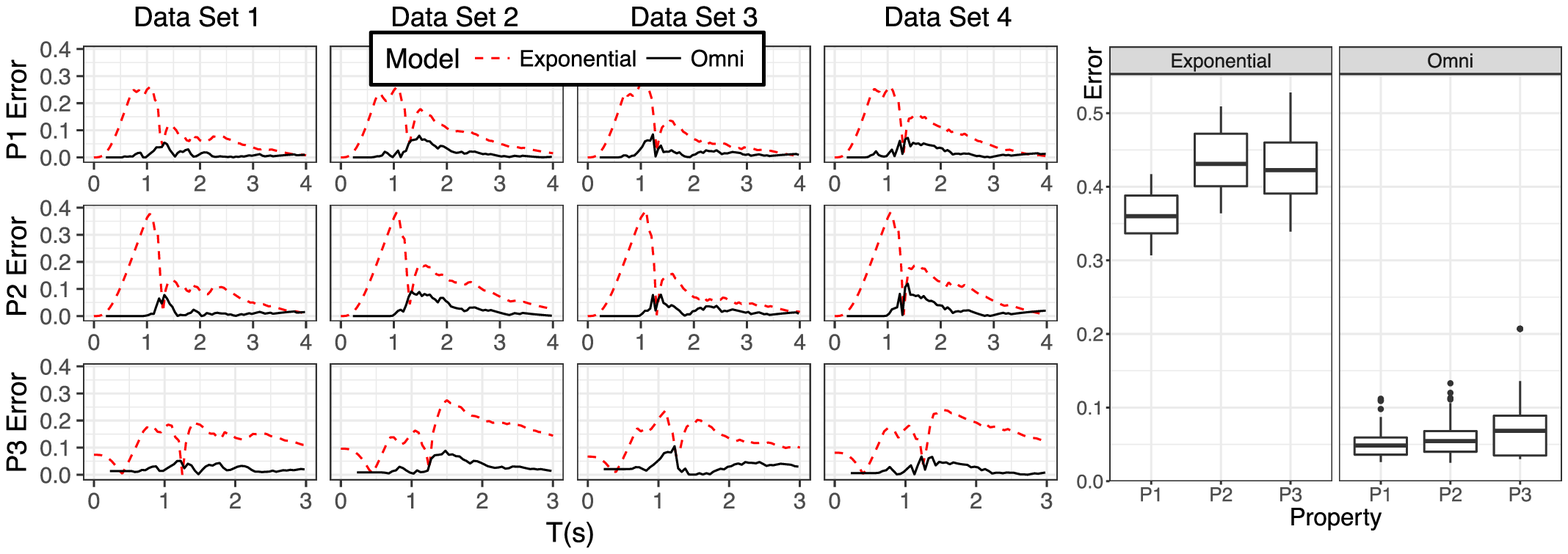}
\caption{Prediction error for the web application properties, for training and testing datasets from different runs 
}
\label{fig:matrixPlotWS}
\end{figure}

\subsection{RQ1 (Accuracy/No overfitting) \label{subsect:rq1}} 

The generation of \acronym-refined CTMC models requires the processing of finite datasets produced by the components of the analysed system, in order to extract key model features. To be useful, these CTMCs should accurately predict the values of the system properties for other system runs, i.e.\ should not be overfitted to the datasets used to generate them.

\subsubsection{Travel Web Application \label{subsubsect:rq1-web}}

To assess whether the  \acronym\ web application models possess this property, we obtained three additional datasets (labelled `Data Set 2', `Data Set 3' and 'Data Set 4') for the travel web application. Each new dataset corresponds to a four-hour run, with all datasets (including the original dataset, `Data Set 1') captured over a period of two days. Fig.~\ref{fig:matrixPlotWS} shows the difference between the property values predicted by the CTMC analysis and the actual property values taken from each of the four datasets. Results are shown for the initial CTMC from Fig.~\ref{fig:CTMC} (labelled `Exponential' in the diagrams) and the refined models obtained using `Data Set~1' (labelled `Omni' in the diagrams). In all cases, the \acronym-refined CTMCs significantly improve the accuracy of the analysis when compared to the traditional CTMC analysis approach.

\subsubsection{IT Support System \label{subsubsect:rq1-IT}}

For the IT support system, \acronym-refined CTMCs for the two properties were produced from half of the available system logs (`Data Set 1') as described in Section~\ref{subsect:second-case-study}. We then assessed the accuracy of the predictions obtained using these refined CTMCs against the actual property values extracted from the training `Data Set 1' and from the test dataset (`Data Set 2') produced from the second half of the system logs. Fig.~\ref{fig:matrixPlot} shows the error when the predicted values are compared to actual values for the two datasets. For both datasets, the \acronym-refined CTMCs produce results which significantly outperform the results obtained by analysing the high-level CTMC. 

\begin{figure}
\centering
\includegraphics[width=0.95\linewidth, trim=0 0cm 6.95cm 0,clip=true]{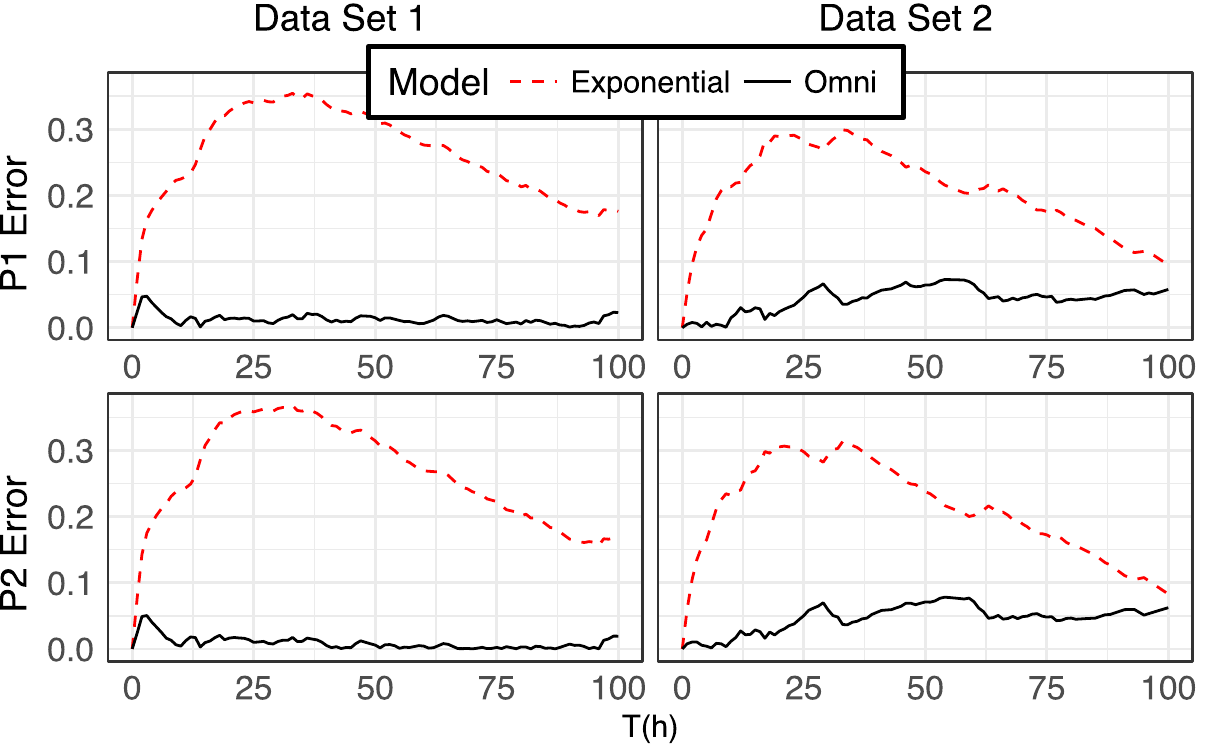}

\caption{Prediction error for the IT support system properties, for training and testing datasets from different three-month time periods 
}
\label{fig:matrixPlot}
\end{figure}

\subsubsection{Discussion}

The experiments described in the previous sections show that \acronym\ consistently outperformed the traditional CTMC modelling and analysis approach in both case studies, irrespective of the choice of training set. \revised{To confirm that \acronym\ delivers error reductions reliably, we performed additional experiments in which the training and testing datasets were drawn randomly from all available observations. Thirty experiments were carried out for each of our two systems, and the results displayed similar error reductions to those presented in Sections~\ref{subsubsect:rq1-web} and~\ref{subsubsect:rq1-IT}.} This shows that \acronym\ models can effectively predict QoS property values for other system runs than the one used to collect the training datasets employed in the refinement.

\revised{Our additional experiments also showed that the error profiles from Figs.~\ref{fig:matrixPlotWS} and~\ref{fig:matrixPlot} capture several general features for the type of QoS analysis improved by \acronym:
\squishlist
\item[1.]  The initial peak in the `Exponential' prediction error for properties \textbf{P1} and \textbf{P2} from Fig.~\ref{fig:matrixPlotWS} is characteristic of the inability of exponential distributions to model delays, as also explained in Section~\ref{subsect:overview}. \acronym\ does not suffer from this limitation. Note that this modelling error does not affect properties \textbf{P1} and \textbf{P2} from Fig.~\ref{fig:matrixPlot} because the delays for the IT support system are insignificant compared to the holding times. 
\item[2.] The second peak in the `Exponential' prediction error for properties \textbf{P1} and \textbf{P2} from Fig.~\ref{fig:matrixPlotWS}, and the first peak for properties \textbf{P1} and \textbf{P2} from Fig.~\ref{fig:matrixPlot} are representative of the inability of exponential distributions to model long tails (due to operations occasionally having much longer execution times than their typical execution times). As such, the estimated rates of the exponential distributions are too low, and the predictions are overly conservative. These error peaks are particularly high (above 0.3) for the IT support system, as IT support personnel occasionally required very long times to address a user request. Again, \acronym\ yields much smaller prediction errors around these peaks.
\item[3.] The multiple peaks in the `Exponential' prediction error for property \textbf{P3} from Fig.~\ref{fig:matrixPlotWS} is characteristic of derived properties, i.e., properties defined using multiple ``primitive'' properties (in this case, \textbf{P3} represents profit, and is defined as the difference between revenue and penalties). The multiple peaks are due to the prediction errors for the primitive properties peaking at different time moments. As before, \acronym\ significantly dampens these peaks.
\squishend
}

\subsection{RQ2 (Refinement Granularity)} 

To evaluate the effects of refinement granularity we constructed a set of \acronym\ models by varying:
\squishlist
\item[1)] $k_i$, the number of states in the Erlang delay models from the joint delay modelling of \acronym\ (cf.~Theorem~\ref{th:erlang});
\item[2)] $\alpha$, the PHD model fitting threshold used in the holding time modelling of \acronym\ (cf.~Algorithm~\ref{alg:htmodelling}).
\squishend
Larger values of $k_i$ are associated with increased accuracy in the modelling of delays, whilst reducing $\alpha$ corresponds to finer-grained refinement in  the PHD modelling.

\subsubsection{Travel Web Application}

\begin{table*}

\caption{Effects of the \acronym\ refinement granularity on web application model}
\label{tab:granularityWS}
\centering

\sffamily
\begin{tabular}{ccccccccccc} \toprule
 & & \multicolumn{3}{c}{P1} & \multicolumn{3}{c}{P2}& \multicolumn{3}{c}{P3}\\ \cmidrule{3-11}
 $k_i$ & $\alpha$ & \textbf{\#states} & \textbf{Error} & $T_V(s)$\hspace*{4mm}  &\textbf{\#states} & \textbf{Error} & $T_V(s)$\hspace*{4mm}  &\textbf{\#states} & \textbf{Error} & $T_V(s)$  \\ 
\midrule
\multicolumn{2}{c}{Initial CTMC} & 7 & 0.325 & 1.9 & 7 & 0.402 & 1.9 &7 & 0.377 & 1.9 \\ \midrule
10 & 0.2 & 82 & 0.085 & 3.9 & 45 & 0.126 & 3.5 & 82 & 0.094 & 5.0 \\
100 & 0.2 & 262 & 0.049 & 5.6 & 135 & 0.066 & 4.3 & 262 & 0.077 & 7.7 \\
259 & 0.2 & 580 & 0.045 & 8.8 & 294 & 0.060 & 5.8 & 580 & 0.074 & 12.6 \\ \midrule
10 & 0.1 & 232 & 0.078 & 6 & 118 & 0.112 & 4.3 & 232 & 0.083 & 8.1 \\
100 & 0.1 & 412 & 0.043 & 7.8 & 208 & 0.049 & 5.1 & 412 & 0.063 & 11.0 \\
259 & 0.1 & 730 & 0.038 & 11.4 & 367 & 0.042 & 6.8 & 730 & 0.059 & 16.5 \\ \midrule
10 & 0.05 & 618 & 0.075 & 13.8 & 376 & 0.106 & 7.9 & 618 & 0.081 & 19.6 \\
100 & 0.05 & 798 & 0.041 & 16.0 & 466 & 0.044 & 8.8 & 798 & 0.061 & 22.7 \\
259 & 0.05 & 1116 & 0.036 & 20.6 & 625 & 0.036 & 10.8 & 1116 & 0.057 & 29.4 \\ \bottomrule

\end{tabular} 
\rmfamily

\end{table*}

The experimental results from the web application case study are presented in Table~\ref{tab:granularityWS}. 
As $k_i$ is increased from $10$ to $100$ and from $100$ to $259$, the error is reduced.\footnote{These $k_i$ values are taken from Table~\ref{table:precomputed}.} However, this improvement shows diminishing returns for all properties as $k_i$ becomes large. The same pattern occurs as $\alpha$ is decreased, with smaller errors for smaller $\alpha$ values but only a marginal reduction in error as $\alpha$ is reduced from $0.1$ to $0.05$. 

Since $k_i$ controls the number of states associated with delays, increasing $k_i$ also increases the total number of states associated with the model. The model size also increases as $\alpha$ is decreased. 

Finally, the experimental results confirm that the models for property \textbf{P2} are consistently much smaller than for \textbf{P1} and \textbf{P3} since more states from the initial CTMC are in the ``exclude from refinement'' set $S_\textsf{X}$ when evaluating \textbf{P2} than when evaluating the other properties (cf.\ Table~\ref{tab:classification}).  $T_V$ is the total time for PRISM to verify each property in the interval [0, $T_\mathrm{max}$] with a time step of 0.05s and includes the time taken for model construction. All experiments presented here and throughout the rest of the paper were carried out on a MacBook Pro with 2.9 GHz Intel i5 processor and 16Gb of memory. As the model increases in size, and accuracy improves, the time taken for verification also increases, up to $29.4$s for the finest-grained model used to evaluate property \textbf{P3} across the entire interval [0, $T_\mathrm{max}$].

\subsubsection{IT Support System}

When \acronym\ is applied to the IT support system, the delay threshold of $0.01$~hours chosen as explained in Section~\ref{subsect:second-case-study} means that the delay modelling was omitted (i.e.\ delays were approximated to zero). As such, we were not interested in varying $k_i$ in this case study, and Table~\ref{tab:granularityCS} only shows the effects of decreasing $\alpha$ on the refined models. Like in the first case study, decreasing $\alpha$ gradually reduces the prediction error, with a significant error reduction obtained even for the largest $\alpha$ from our experiments (e.g.\ an over tenfold reduction from $26.3$ for the initial, high-level CTMC and property \textbf{P1} to just $2.45$ for the coarsest-granularity CTMC generated for $\alpha=0.6$). Diminishing returns in terms of error reduction are achieved for property \textbf{P2}; for \textbf{P1}, this trend is not clearly distinguishable for the tested $\alpha$ values.

As expected, the model size grows as $\alpha$ is decreased, leading to a corresponding increase in the verification time $T_V$. $T_V$ includes the time for the construction of the model and for PRISM to analyse the property in the interval $[0, 100\textrm{h}]$ with a time step of one~hour (i.e.\ 100 verification sessions). The largest verification time is $274.6$s for the finest-granularity CTMC obtained for property \textbf{P1}, which is entirely acceptable for an offline verification task. 

During the component classification step of \acronym, the exclusion sets for the two properties are calculated as 
$S_\mathsf{X} = \{s_6,s_7\}$ for \textbf{P1} and $S_\mathsf{X} = \{s_3, s_5,s_6,s_7\}$ for \textbf{P2}. Therefore, the models associated with \textbf{P2} are consistently smaller than those associated with \textbf{P1}, whose exclusion set $S_\mathsf{X}$ contains only two states. 

\begin{table}
\caption{Effects of the OMNI refinement granularity on the IT support system model}
\label{tab:granularityCS}
\centering
\begin{tabular}{cccccccc} \toprule
 & \multicolumn{3}{c}{P1} & \multicolumn{3}{c}{P2}\\ \cmidrule{2-7}
 $\alpha$&\textbf{\#states} & \textbf{Error} & $T_V(s)$\hspace*{4mm}  &\textbf{\#states} & \textbf{Error} & $T_V(s)$  \\ 
\hline
Initial &&&&&&\\
CTMC & 8 & 26.3 & 3.0 & 8 & 26.5 & 3.0 \\ \midrule
0.6 & 44 & 2.45 & 42.8 & 41 & 2.11 & 41.7 \\
0.4 & 61 & 2.24 & 48.2 & 51 & 1.74 & 44.4\\
0.2 & 202 & 1.39 & 95.8 & 174 & 0.95 & 84.4  \\
0.1 & 329 & 1.27 & 139.6 & 259 & 0.87 & 110.6  \\
0.05 & 722 & 1.01 & 274.6 & 346 & 0.84 & 139.8  \\ \bottomrule
\end{tabular} 
\end{table}

\subsubsection{Discussion}

For both case studies and all considered QoS properties, considerable improvements in model accuracy are obtained even with small, coarse-grained \acronym\ models. As such, \acronym\ can offer significant improvements in accuracy over traditional CTMC modelling techniques even when computational resources are at a premium. Additional, but typically diminishing, gains in prediction accuracy are obtained through increasing the granularity of the refinement. Expectedly, this leads to a corresponding increase in verification time. For our two systems, this time did not exceed 10 minutes (and was typically much smaller) for all considered properties and model granularities -- an acceptable overhead for the offline verification task performed by \acronym. 

\subsection{RQ3 (Training dataset size) \label{subsubsect:rq3}}

In both case studies, we ran a set of experiments to evaluate the effect of reducing the training dataset size on the accuracy of \acronym\ models. For each experiment, training subsets were constructed by randomly selecting a percentage of all available datasets used to answer the previous research questions. The sizes of these selected subsets were 80\%, 60\%, 40\% and 20\% of the complete training dataset from Sections~\ref{subsubsect:rq1-web} and~\ref{subsubsect:rq1-IT}. For each system and each of its analysed QoS properties, the experiments were repeated 30 times, with the property errors recorded.

\subsubsection{Travel Web Application}

Table~\ref{table:RQ3_WS} shows the mean error and standard deviation (labelled `sd') for the web application case study with $k_i = 259$ and $\alpha=0.1$. As the training dataset size decreases, the error and standard deviation associated with each property show an increasing trend. However, we note that at 80\% the prediction errors show little difference to the 100\% figures -- the mean errors at 100\% are very close to the 80\% errors, and well within one standard deviation of the 80\%~mean. This suggests that 80\% of the complete dataset is sufficient to capture the characteristics of the underlying component distributions for this case study.

\begin{table}
\caption{Web application -- training dataset size effect on prediction accuracy, shown as average error and standard deviation over 30~runs}
\label{table:RQ3_WS}
\centering
\begin{tabular}{p{1.55cm}p{1.8cm}p{1.8cm}p{1.8cm}} 
\toprule
\hspace*{-1mm}\textbf{Dataset$^\dagger$} & $\!$\textbf{P1 Error} & $\!$\textbf{P2 Error} & $\!$\textbf{P3 Error} \\ 
\midrule
\hspace*{-1mm}100\%$^{\dagger\dagger}$ & $\!$0.038  sd N/A$^*$   & $\!$0.042 sd N/A$^*$ & $\!$0.059 sd N/A$^*$\\
\hspace*{-1mm}80\% & $\!$0.038 sd 0.006 & $\!$0.043 sd 0.005 & $\!$0.058 sd 0.023 \\
\hspace*{-1mm}60\% & $\!$0.046 sd 0.013 & $\!$0.048 sd 0.014 & $\!$0.078 sd 0.041 \\
\hspace*{-1mm}40\% & $\!$0.057 sd 0.017 & $\!$0.063 sd 0.022 & $\!$0.105 sd 0.054 \\
\hspace*{-1mm}20\% & $\!$0.083 sd 0.032 & $\!$0.076 sd 0.026 & $\!$0.156 sd 0.075\\
\midrule
\hspace*{-1mm}Initial CTMC & $\!$0.325 sd N/A$^*$ & $\!$0.402 sd N/A$^*$ & $\!$0.377 sd N/A$^*$ \\
\bottomrule
\multicolumn{4}{l}{$^\dagger$Percentage of complete 270-element training dataset}\\ 
\multicolumn{4}{l}{$^{\dagger\dagger}$Single run using entire data set}\\
\multicolumn{4}{l}{$^*$Single run, so no standard deviation}
\end{tabular}
\end{table}

\subsubsection{IT Support System}

For the IT support system, the experimental results are provided in Table~\ref{table:RQ3_CS}. As for the other case study, we observe that reducing the size of the training sets leads to a trend where the prediction error and the standard deviation increases. This also happens when the size of the training dataset is reduced from 100\% to 80\%, suggesting that additional slight improvements may be possible by further increasing the size of the initial training dataset. 

\begin{table}
\centering
\caption{IT support system -- training dataset size effect on prediction accuracy, shown as average error and standard deviation over 30~runs}
\label{table:RQ3_CS}
\begin{tabular}{p{1.8cm}p{1.8cm}p{1.8cm}} 
\toprule
\hspace*{-1mm}\textbf{Dataset$^\dagger$} & $\!$\textbf{P1 Error} & $\!$\textbf{P2 Error} \\ 
\midrule
\hspace*{-1mm}100\%$^{\dagger\dagger}$ & $\!$1.39  sd N/A$^*$   & $\!$0.95 sd N/A$^*$\\
\hspace*{-1mm}80\% & $\!$1.53 sd 0.774 & $\!$1.35 sd 0.570\\
\hspace*{-1mm}60\% & $\!$1.57 sd 0.942 & $\!$1.55 sd 0.797\\
\hspace*{-1mm}40\% & $\!$2.37 sd 1.450 & $\!$2.19 sd 1.396\\
\hspace*{-1mm}20\% & $\!$3.70 sd 2.825 & $\!$3.81 sd 3.039\\
\midrule
\hspace*{-1mm}Initial CTMC & $\!$26.31 sd N/A$^*$ & $\!$26.46 sd N/A$^*$ \\
\bottomrule
\multicolumn{3}{l}{$^\dagger$Percentage of complete 705-element training dataset}\\
\multicolumn{3}{l}{$^{\dagger\dagger}$Single run using entire data set}\\
\multicolumn{3}{l}{$^*$Single run, so no standard deviation}
\end{tabular}
\end{table}

\subsubsection{Discussion}

For both case studies we note that even modest training dataset sizes show a significant improvement over the traditional approach to CTMC-based analysis of QoS properties. For the web application, a training set consisting of 20\% of the original dataset equates to only 54 request handling observations, and reduces the mean estimation error by between 50--81\% for the properties of interest. For the IT system, 20\% of the original training dataset equates to 141~tickets processed, with the processing of only five tickets using the $\mathsf{addInfo}$ component of the system, yet the prediction errors for $\textbf{P1}$  and $\textbf{P2}$ are both reduced by approximately 86\%.

\begin{table*}
\caption{Web application -- comparison of \acronym\ with the preliminary CTMC refinement approach from~\cite{OMNIref} \label{tab:RQ4_WS}}
\centering
\begin{tabular}{p{5cm}ccccccccc}
\toprule
& \multicolumn{3}{c}{\textbf{P1}}& \multicolumn{3}{c}{\textbf{P2}}& \multicolumn{3}{c}{\textbf{P3}} \\ \cmidrule{2-10}
\textbf{Model}        &\textbf{\#states} & \textbf{Error} & $T_V(s)$\hspace*{4mm}  &\textbf{\#states} & \textbf{Error} & $T_V(s)$ \hspace*{4mm} &\textbf{\#states} & \textbf{Error} & $T_V(s)$ \\
\midrule
Preliminary OMNI ($\alpha=0.1$, $k_i=259$)  & 1766   & 0.037 & 24.95\hspace*{4mm}     & 1766   & 0.039 &  23.79\hspace*{4mm}      & 1766   & 0.063 & 38.21 \\ \midrule
\acronym\ ($\alpha=0.1$, $k_i=259$)  & 730    & 0.038 & 11.40\hspace*{4mm}     & 367    & 0.042 &  6.80\hspace*{4mm}       & 730    & 0.059 & 16.5 \\ 
\midrule
\acronym\ ($\alpha=0.05$, $k_i=259$)& 1116 & 0.036 & 20.6\hspace*{4mm} & 625 & 0.036 & 10.8\hspace*{4mm} & 1116 & 0.057 & 29.4 \\ \midrule
High-level CTMC & 7 & 0.325 & 2.53\hspace*{4mm} & 7 & 0.402 & 2.50\hspace*{4mm} & 7 & 0.377 & 2.47 \\ 
\bottomrule
\end{tabular}
\end{table*}

\begin{table*}

\vspace*{3mm}
\caption{IT support system -- comparison of \acronym\ with the preliminary CTMC refinement approach from~\cite{OMNIref} \label{tab:RQ4_CS}}
\centering
\begin{tabular}{>{\raggedright}p{5cm}cccccc} 
\toprule
& \multicolumn{3}{c}{P1}& \multicolumn{3}{c}{P2} \\ 
\cmidrule{2-7}
 \textbf{Model}          &\textbf{\#states} & \textbf{Error} & $T_V(s)$\hspace*{4mm}  &\textbf{\#states} & \textbf{Error} & $T_V(s)$ \hspace*{4mm} \\ 
\midrule
Preliminary OMNI ($\alpha=0.2$, $k_i=10$)  & 265   & 1.39 & 261.5    & 265  & 0.95 & 239.6   \\ \midrule
\acronym\ ($\alpha=0.2$)    & 202   & 1.39 &  95.8    & 174  & 0.95 &  84.4 \\ 
\midrule
High level CTMC & 8 & 26.3 & 3.0 & 8 & 26.5 & 3.0 \\ 
\bottomrule
\end{tabular}
\end{table*}

\subsection{RQ4 (Component classification) \label{subsect:rq4}} 

We evaluated the effects of extending our preliminary CTMC-refinement method from~\cite{OMNIref} with the component classification step described in Section~\ref{sect:comp-class}. To this end, we performed experiments to compare the model size, verification time and accuracy of the refined CTMCs generated by the \acronym\ method described in this paper and of the refined CTMCs produced by our preliminary method, which refines each system component independently, irrespective of whether it impacts the analysed QoS property or not.

\subsubsection{Travel Web Application}

For the web application, refined CTMCs were built using first the preliminary \acronym\ method from~\cite{OMNIref} and the fully fledged version of \acronym\ from this paper, initially with parameters $k_i=259$ and $\alpha=0.1$. The first two rows from Table~\ref{tab:RQ4_WS} summarise these experimental results, which show that the use of component classification yields significant reductions in the number of model states and the verification time for all three properties compared to the preliminary method. As expected given the CTMC state partition from Table~\ref{tab:classification} and \acronym\ rules from Table~\ref{table:refinementRules}, the largest reductions are achieved for property \textbf{P2} (79\% fewer model states, and 71\% shorter verification time). For properties $\textbf{P1}$ and $\textbf{P3}$, the use of component classification led to a reduction in model size of over 58\%, and to reductions in verification time of 54\% and 57\%,  respectively. 

The prediction errors are very close for both \acronym\ variants, and considerably smaller than the errors for the high-level CTMC (provided in the last row of Table~\ref{tab:RQ4_WS} for convenience). However, the errors are negligibly larger when component classification is used. This is due to the use of fewer states for \acronym's joint delay modelling compared to the separate modelling of component delays in~\cite{OMNIref}. As shown in Table~\ref{tab:granularityWS}, additional reductions in prediction error may be achieved by increasing $k_i$ or reducing $\alpha$ if required. For example, by setting $\alpha =0.05$, the refined models still have much fewer states than the preliminary \acronym\ model, show an improvement in verification time of between 17--54\%, and are more accurate. The third row of Table~\ref{tab:RQ4_WS} shows again these experimental results for ease of comparison.

\subsubsection{IT Support System}

For the second case study, the experimental results are presented in Table~\ref{tab:RQ4_CS}. Whilst our fully fledged \acronym\ allows for component delays to be omitted from the refinement when they are below a delay threshold (cf.\ Section~\ref{section:tool}), this was not possible in our previous work. Therefore, to ensure a fair comparison, we used a small $k_i$ value ($k_i=10$) when generating refined CTMCs with the preliminary \acronym\ variant. As shown by the experimental results, using the fully fledged \acronym\ yields smaller refined models that take 74\% and 64.7\% less time to verify for the IT system properties \textbf{P1} and \textbf{P2}, respectively. Furthermore, these smaller refined models achieve the same prediction accuracy as the larger models generated by the preliminary \acronym.

\subsubsection{Discussion}

The \acronym\ method described in this paper includes a component classification step in its model construction. Since this step uses the high-level CTMC model only, the time taken for its execution is very small. For the web application, the time taken to classify the high-level CTMC states for all three properties was 2.3s, and for the IT support system this step took only 1.8s. 

For both case studies presented we have shown that for all the properties considered it was possible to generate \acronym\ models which are smaller, faster to verify and no less accurate than those produced by our preliminary approach from~\cite{OMNIref}. The amount of verification time saved depends on the number of states for which delays can be combined, and on the number of states which can be excluded from refinement -- but these savings were considerable in all our experiments with the two real-world systems.


\begin{table*}
\caption{\revised{Characteristics of the case studies used to evaluate \acronym\ \label{table:case-study-diffs}}}
\centering
\revised{\begin{tabular}{p{2.8cm}p{6.8cm}p{6.8cm}}
\toprule
\textbf{Characteristic} & \textbf{Case study 1} & \textbf{Case study 2} \\
\midrule
\multicolumn{3}{c}{\cellcolor{gray!25}\textbf{System}}\\
\midrule
Type of system & Prototype service-based system developed by the \acronym\ team & Production IT support system developed by, running at, and managed by university in Brazil\\
\midrule
System components & Mix of six high-performance (commercial) and budget (free) third-party web services invoked remotely over the Internet & Proprietary software components deployed on university computing infrastructure, and supporting human tasks with high variance in temporal characteristics \\
\midrule
Size of system & $3188$ lines of code & $1276131$ lines of code \\
\midrule
Component execution times & Tens to hundreds of milliseconds & Minutes to hours \\
\midrule
Operational profile & Assumed values for the probabilities of the different types of requests & Probabilities of different operation outcomes extracted from the real system logs\\
\midrule
\multicolumn{3}{c}{\cellcolor{gray!25}\textbf{Datasets}}\\
\midrule
Dataset source & Obtained from invocations of real web services & Taken from actual system logs \\ 
\midrule
Key dataset features & Significant delays (compared to holding times) due to network latency & Long tails and outliers due to a small number of complex user tickets; multi-modal response times and regions of zero density due to different experience levels of IT support personnel; negligible delays \\
\midrule 
\multicolumn{3}{c}{\cellcolor{gray!25}\textbf{Analysed QoS properties}}\\
\midrule
Types of properties & Overall success probability (\textbf{P1}) & Overall response time (\textbf{P1}) \\
& Success probability for ``day-trip'' requests (\textbf{P2}) & Response time for ``straightforward'' user tickets (\textbf{P2})\\
& Profit $=$ revenue $-$ penalties (\textbf{P3}) &\\
\midrule 
\multicolumn{3}{c}{\cellcolor{gray!25}\textbf{Purpose of QoS analysis}}\\
\midrule
Supported stage of & Design of new system & Verification of existing system \\
development process & & \\
\bottomrule
\end{tabular}}
\end{table*}

\section{Threats to validity}
\label{sect:threats}

\subsection{External validity}

External validity threats may arise if the stochastic characteristics of the systems from our case studies are not indicative of the characteristics of other systems. To mitigate this threat, we used two significantly different systems from different domains 
\revised{for the \acronym\ evaluation. The section labelled `System' from Table~\ref{table:case-study-diffs} summarises the multiple characteristics that differ between these systems.} 

\revised{In addition, the datasets used in the two case studies present different characteristics, as shown in the `Datasets' section from Table~\ref{table:case-study-diffs}. In particular,} the datasets for the service-based system were obtained from real web services, while for the 
IT support system they were taken from the actual system logs. This gives us confidence that the stochastic characteristics of the two systems (including regions of zero density, multi-modal response times, and long tails) are representative for many real-world systems. 

\revised{The next section from Table~\ref{table:case-study-diffs} summarises the different types of QoS properties analysed in our case studies. The transient fragment of continuous stochastic logic (whose analysis accuracy is improved by \acronym\ supports, cf.\ Section~\ref{subsect:CSL}) supports the specification of multiple classes of QoS properties of interest, including success probability, profit/cost and response time, and our case studies considered examples of all of these.}

\revised{Finally, as shown in the `Purpose of QoS analysis' section from Table~\ref{table:case-study-diffs}, the first case study applied \acronym\ during the design stage of the development process, whereas the second case study assessed \acronym\ by verifying QoS properties of an existing system.}

Another external threat may arise if the \acronym-refined CTMC models were too large to be verified within a reasonable amount of time. The \acronym\ approach mitigates this threat by allowing for the refinement to be carried out at different levels of granularity, and our experiments indicate that significant improvements in prediction accuracy is achievable with modest enlargement of the models. 

\begin{figure*}
\centering
\hspace*{-3mm}\mbox{\includegraphics[width=0.22\linewidth]{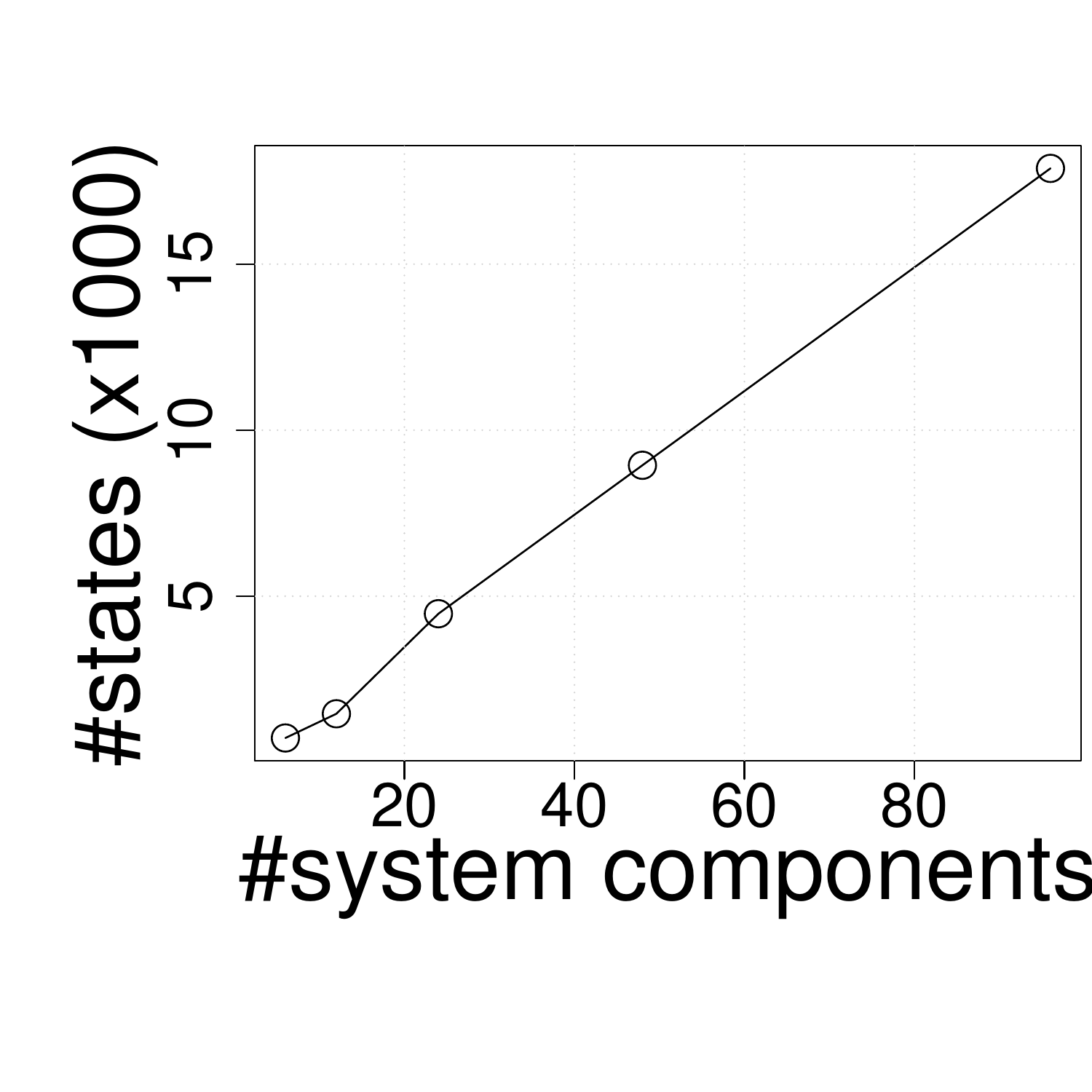}\hspace*{7mm}
\includegraphics[width=0.22\linewidth]{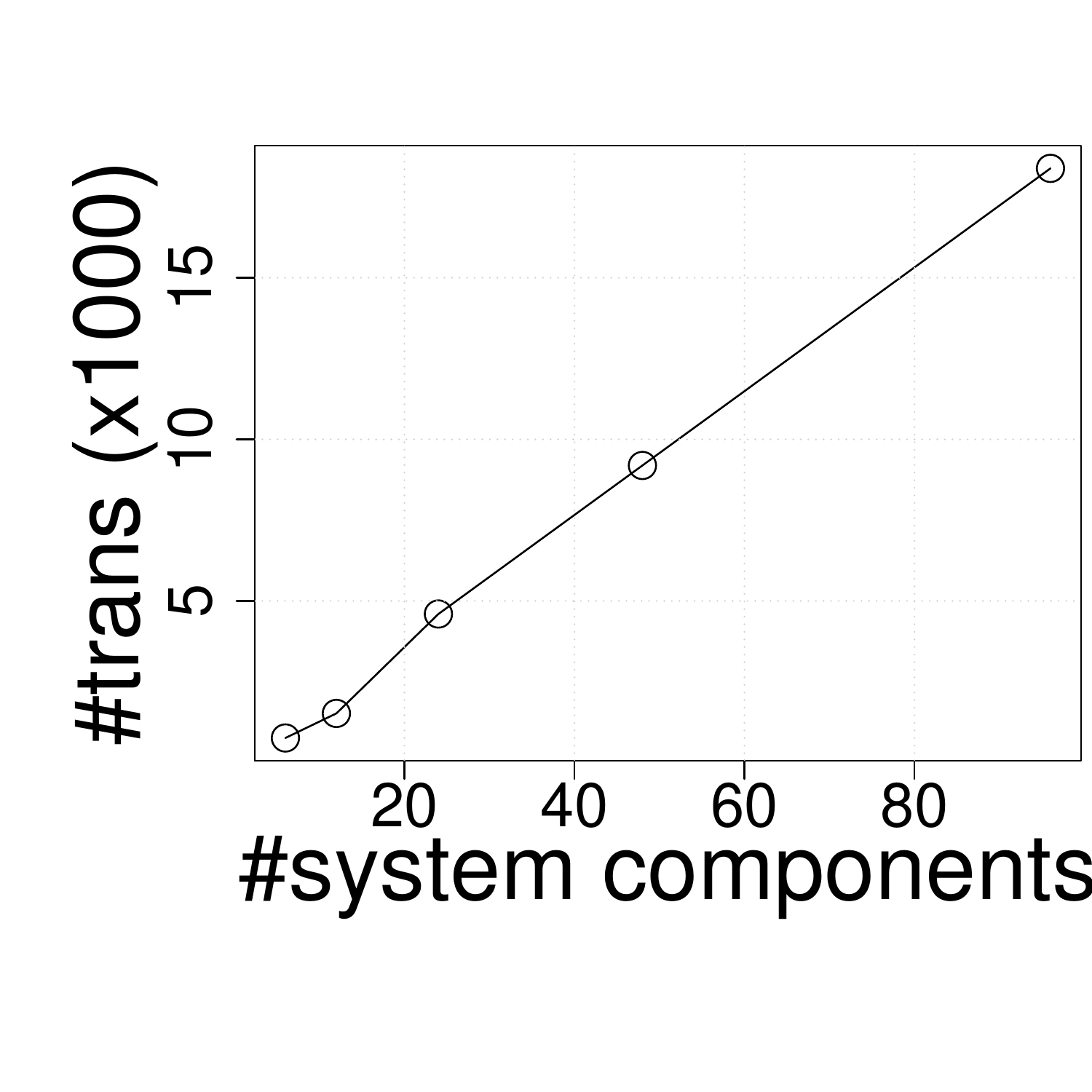}\hspace*{7mm}
\includegraphics[width=0.22\linewidth]{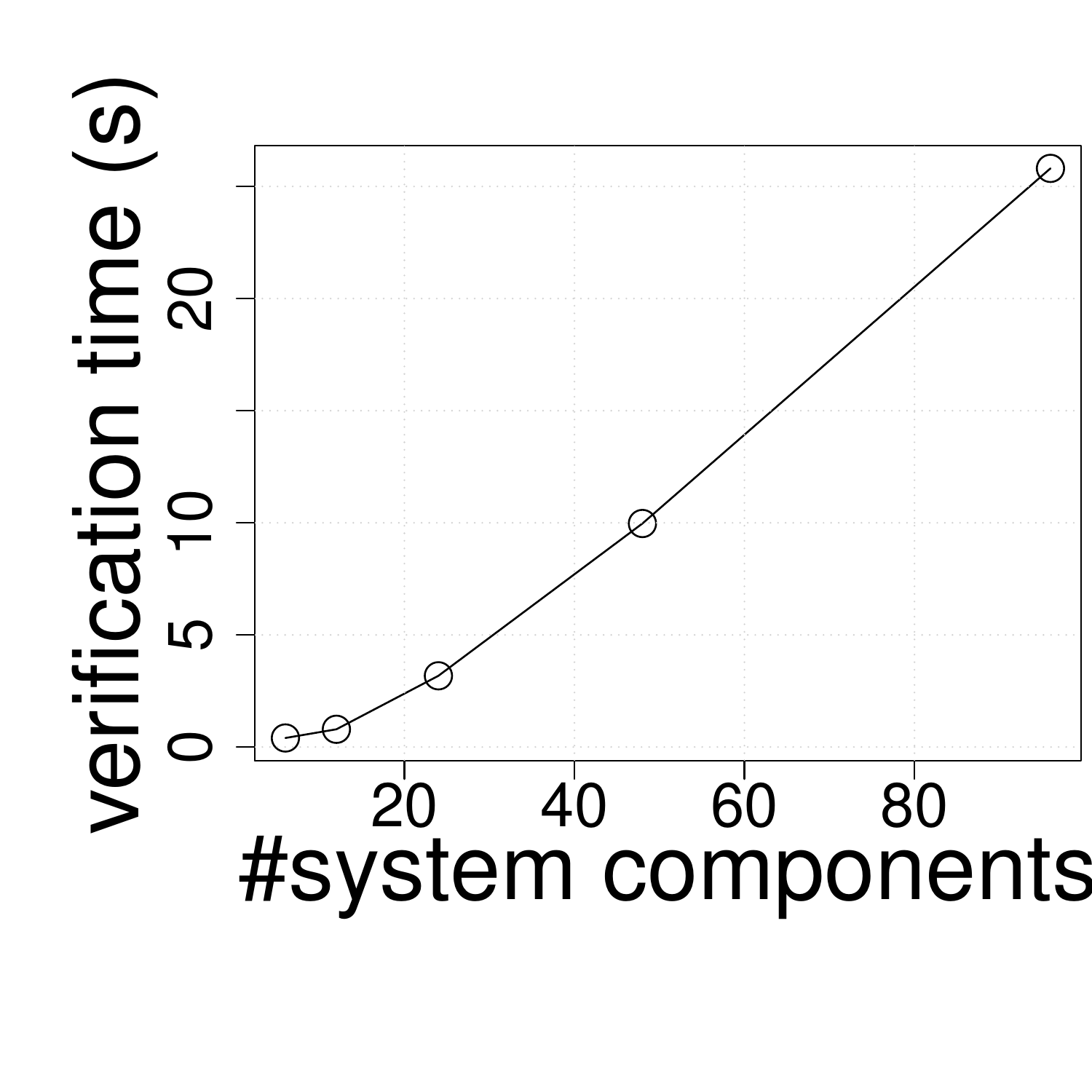}}

\vspace*{-5mm}
\caption{Refined CTMC states, transitions and verification time \revised{for property \textbf{P1} at a single time point,} for system sizes up to 16 times larger than the web application}
\label{fig:LargerModels}
\end{figure*}

Our two case studies are based on real systems, but these systems have a relatively small number of \emph{\revised{modelled} components}. \revised{For the large IT system from the second case study, the small number of \emph{modelled components} is due to the inclusion in the model of only components that influence the analysed QoS properties -- a modelling technique termed \emph{abstraction}.} 

\revised{For systems with larger numbers of modelled components, we note} that the increase in model size due to the \acronym\ refinement is only linear in the number of system components. Moreover, as \acronym\ uses acyclic PHDs, the number of transitions also increases linearly. Modern model checkers can handle CTMCs with $10^5 - 10^6$ states~\cite{Jansen2008} and as such we expect \acronym\ to scale well with much larger systems. We confirmed these hypotheses by constructing models with 12, 24, 48 and 96 components by combining 2, 4, 8 and 16 instances of our web application CTMC from Fig.~\ref{fig:CTMC}.\footnote{We did not perform similar experiments for the IT support system as they would not have been qualitatively different.} \acronym\ was then used to refine the composite models with $k_i\!=\!259$ and $\alpha\!=\!0.1$. For each refined CTMC, we measured the number of states, the number of transitions, and the time taken to verify property $\textbf{P1}$ of the travel web application at \revised{the single time point} $T\!=\!20s$ (since $\textbf{P2}$ and $\textbf{P3}$ can not be meaningfully extrapolated to these larger systems). The results of these experiments, shown in Fig.~\ref{fig:LargerModels}, confirm the predicted linear increase in the verification overhead with the system size.

\subsection{Construct validity}

Construct validity threats may be due to the assumptions made when collecting the datasets or when defining the QoS properties for our model refinement experiments. To address the first threat, we collected the datasets from a real IT support system and from a prototype web application that we implemented using standard Java technologies and six real web services from three different providers. For the first system, the datasets were collected over a period of six months, and for the second system they were collected on two different days, at different times of day. Furthermore, we used different datasets for training and testing. To mitigate the second threat, we analysed three performance and cost properties of web application, and two typical performance properties of the IT system.

\subsection{Internal validity}

Internal validity threats can originate from the stochastic nature of the two analysed systems or from bias in our interpretation of the experimental results.  We addressed these threats by provided formal proofs for our CTMC refinement method, by reporting results from multiple independent experiments performed for different values of the \acronym\ parameters, and by analysing several QoS properties at multiple levels of refinement granularity. Additionally, we made the experimental data and results publicly available on our project webpage in order to enable the replication of our results.


\section{Related work \label{section:related}}

To the best of our knowledge, \acronym\ is the first tool-supported method for refining high-level CTMC models of component-based systems based on separate observations of the execution times of the system components. 

\acronym\ builds on recent approaches to using PHDs to fit non-parametric distributions, a research area that has produced many efficient PHD fitting algorithms over the past decade \cite{15326340701300712,doi:10.1080/03610918.2013.848895,Wang2008,doi:10.1080/03610926.2010.483306}. 
Buchholz et al.~\cite{buchholz2014phase} and Okamura and Dohi \cite{Okamura2016} present overviews of the theory and applications of PHDs in these types of analysis, in domains including the modelling of call centres~\cite{ishay2002fitting} and healthcare processes~\cite{fackrell2009modelling,marshall2009simulating}. However, these algorithms and applications consider the distribution of timing data for a complete end-to-end process rather than separate timing datasets for the components of a larger system as is the case for \acronym. This focus on a single dataset also applies to the cluster-based PHD fitting method from~\cite{reinecke2012cluster} and its implementation within the efficient PHD-fitting tool HyperStar~\cite{reinecke2012hyperstar}, which \acronym\ uses for its holding-time modelling. 

Recent work by Karmakar and Gopinath~\cite{karmakar2015markov} has shown that PHD models can be used in conjunction with CTMC solvers to verify storage reliability models. In this work, Weibull distributions are assumed to more accurately describe the processes of concern, and PHDs are used to approximate these distributions. The PRISM probabilistic model checker is then used to assess properties concerned with system reliability. Unlike this approach, \acronym\ is applicable to the much wider class of problems where additional QoS properties need to be analysed and where the relevant component features correspond to non-parametric distributions that cannot be accurately modelled as Weibull distributions.

The analysis of non-Markovian processes using PHDs is considered in~\cite{ciobanu2014phase}, where a process algebra is proposed for use with the probabilistic model checker PRISM. However, \cite{ciobanu2014phase} presents only the analysis of a simple system based on well-known distributions, and does not consider PHD fitting to real data nor how its results can be exploited in the scenarios tackled by \acronym.

To address the significant difficulties that delays within a process pose to PHD fitting, Koren{\v{c}}iak et al.\  \cite{krvcal2014dealing} have tackled probabilistic regions of zero density by using interval distributions to separate discrete and continuous features of distributions. Similar work~\cite{brazdil2015optimizing} supports the synthesis of timeouts in fixed-delay CTMCs by using Markov decision processes. Unlike \acronym,~\cite{brazdil2015optimizing} and \cite{krvcal2014dealing} do not consider essential non-Markovian features of real data such as multi-modal and long-tail distributions, and thus cannot handle empirical data that has these common characteristics.

Finally, non-PHD-based approaches to combining Markov models with real data range from using Monte Carlo simulation to analyse properties of discrete-time Markov chains with uncertain parameters \cite{Meedeniya2014} to using semi-Markov chains to model holding times governed by general distributions \cite{lopez2001}. However, none of these approaches can offer the guarantees and tool support provided by \acronym\ thanks to its exploitation of established CTMC model checking techniques.


\section{Conclusion \label{section:conclusion}}

We presented \acronym, a tool-supported method for refining the CTMC models of component-based systems using observations of the component execution times. To evaluate \acronym, we carried out extensive experiments within two case studies from different domains. The experimental results show that \acronym-refined models support the analysis of transient QoS properties of component-based systems with greatly increased accuracy compared to the high-level CTMC models typically used in software performance engineering. Furthermore, we showed that significant accuracy improvements are achieved even for small training datasets of component observations, and for \acronym\ parameters corresponding to coarse-granularity refinements (and thus to relatively modest increases in model size). 

In our future work, we plan to extend the applicability of \acronym\ to systems comprising components whose behaviour changes during operation. 
This will require \acronym\ to continually refine the CTMC models of these systems based on new component observations. 
We envisage that this extension will enable the runtime analysis of QoS properties for rapidly evolving systems \cite{DBLP:books/sp/ess14/MullerV14} and will support the dynamic selection of new configurations for self-adaptive software used in safety-critical and business-critical applications \cite{8008800}. In addition, we intend to examine the effectiveness of \acronym\ in other application domains, and its ability to estimate a broader range of distributions for the execution times of the system components. 

\bibliographystyle{elsarticle-num}
\bibliography{Omni2}

\end{document}